\documentclass[journal,twoside,romanappendices]{IEEEtran}
\pdfoutput=1
\usepackage{tikz}
\usepackage{pgfplots}
\usepackage{subfigure}
\usepackage{graphicx}
\usepackage{epstopdf}
\usepackage{amsmath}
\usepackage{bm}
\usepackage{array}
\usepackage{amssymb}
\usepackage{multirow}
\usepackage{enumerate}
\usepackage{algorithm, algorithmicx}
\usepackage{algpseudocode}
\usepackage{xcolor}
\usepackage{cite}
\usepackage{soul}
\usepackage{outlines}
\usepackage{pgfplots}
\usetikzlibrary{spy}
\soulregister\cite7
\soulregister\citep7
\soulregister\citet7
\soulregister\ref7
\soulregister\pageref7
\pgfplotsset{compat=1.14}
\newcommand{\red}{\textcolor[rgb]{1.00,0.00,0.00}}
\newcommand{\bl}{\textcolor[rgb]{0.00,0.00,1.00}}

\newtheorem{thm}{Theorem}

\newtheorem{Lemma}{Lemma}

\newtheorem{Rem}{Remark}
\newtheorem{Prop}{Proposition}
\newenvironment{proof}{{\indent \quad \it Proof:}}{\hfill $\blacksquare$\par}

\begin{document}

\title{Performance Analysis and Optimization of Network-Assisted Full-Duplex Systems under Low-Resolution ADCs}
\author{Xiangning~Song, Zhenhao~Ji, Jiamin~Li,~\IEEEmembership{Member,~IEEE}, Pengcheng~Zhu, 
        Dongming~Wang,~\IEEEmembership{Member,~IEEE}, 
        and~Xiaohu~You,~\IEEEmembership{Fellow,~IEEE}\\% <-this % stops a space

\thanks{
\textit{Corresponding authors: Jiamin~Li (email: jiaminli@seu.edu.cn)}

This work was supported in part by the National Key R\&D Program of China under Grant 2021YFB2900300, and by the National Natural Science Foundation of China (NSFC) under Grants 61971127, 61871465, 61871122.

The authors are with National Mobile Communications Research Laboratory, Southeast University, Nanjing, {\rm 210096}, China (email: \{xnsong, zhenhao\_ji, jiaminli, p.zhu, wangdm, xhyu\}@seu.edu.cn). J. Li, D. Wang, and X. You are also with Purple Mountain Laboratories, Nanjing {\rm 211111}, China.}}

\markboth{IEEE SYSTEM JOURNAL,~2022}%
{X. Song \MakeLowercase{\textit{et al.}}: Performance Analysis and Optimization of Network-Assisted Full-Duplex Systems under Low-Resolution ADCs}

\maketitle

\begin{abstract}
Network-assisted full-duplex (NAFD) distributed massive multiple input multiple output (M-MIMO) enables the in-band full-duplex with existing half-duplex devices at the network level, which exceptionally improves spectral efficiency. This paper analyzes the impact of low-resolution analog-to-digital converters (ADCs) on NAFD distributed M-MIMO and designs an efficient bit allocation algorithm for low-resolution ADCs. The beamforming training mechanism relieves the heavy pilot overhead for channel estimation, which remarkably enhances system performance by guiding the interference cancellation and coherence detection. Furthermore, closed-form expressions for spectral and energy efficiency with low-resolution ADCs are derived. The multi-objective optimization problem (MOOP) for spectral and energy efficiency is solved by the deep Q network and the non-dominated sorting genetic algorithm II. The simulation results corroborate the theoretical derivation and verify the effectiveness of introducing low-resolution ADCs in NAFD distributed M-MIMO systems. Meanwhile, a set of Pareto-optimal solutions for ADC accuracy flexibly provide guidelines for deploying in a practical NAFD distributed M-MIMO system.
\end{abstract}

%\IEEEpeerreviewmaketitle

\section{Introduction}\label{sec:Int}
\IEEEPARstart{W}{ith} the introduction of big-data communication scenarios, the swift growth of wireless communication data puts forward urgent demand and brings severe challenges to the existing wireless systems on the channel capacity, latency, and quality of service (QoS) \cite{agiwal2016next}. Huh \emph{et al.} \cite{huh2012achieving} demonstrated the potential gain of massive multiple input multiple output (M-MIMO) on improving spectral efficiency (SE) and acceptable performance with low-complexity receiving and beamforming schemes such as {maximum ratio transmission and combining (MRT/MRC) and zero-forcing transmission and receiving (ZFT/ZFR)}, which provides the possibility for orders of magnitude more data transmission. For the geometric architecture, the distributed antenna system has emerged as a promising architecture for its proximity gains and enhanced macro-diversity\cite{you2010cooperative} with multiple separately-located remote antenna units (RAUs). {However, the high mobility of terminal devices and the heterogeneity of the network structure result in profound asymmetry flow between uplink (UL) and downlink (DL), while the sheer scale of data traffic entails an intolerable consumption of hardware resources.}

\subsection{Related Works and Motivation}
To meet the higher requirement for flexible duplex and asymmetric data traffic, Wang \emph{et al.} \cite{wang2019performance} introduced the conceptual description of network-assisted full-duplex (NAFD) system, which achieves in-band full-duplex with the half-duplex devices and promotes flexible traffic control by dynamically allocating the UL and DL RAUs, i.e., different RAUs have different operational modes, enabling DL transmission and UL reception simultaneously on the same frequency band. Compared to the co-frequency co-time full-duplex (CCFD), NAFD introduces a geographic separation between RAUs and thus weakens the self-interference of intra-RAUs. Nevertheless, the cross-link interference (CLI) via the inter-RAUs UL-to-DL channel seriously affects the system performance and needs to be eliminated. Moreover, channel reciprocity is generally adopted for the classical time division duplex to mitigate the heavy burden of DL training and eliminate the channel state information (CSI) error. In contrast, for the cell-free MIMO, few RAUs serve a given user in a relatively wide area, which immensely weakens the channel-hardening phenomenon \cite{li2018benefits}. Inspired by the beamforming training scheme, the potential benefits of DL pilot estimation become more lavish compared to the statistical CSI in the light channel-hardening scenario \cite{interdonato2019downlink}. Further, Li. \emph{et al.} \cite{li2020network} verified the beamforming training in NAFD distributed M-MIMO with tractable closed-form SE. However, when NAFD distributed M-MIMO employs a perfect quantization scheme, intolerable energy consumption severely limits its future deployment.

To obtain the bonus of M-MIMO, analog-to-digital converters (ADCs) face new challenges in the evolution of the wireless network, where the entire system suffers from the massive processing capacity by power-hungry ADCs with high quantization bits. Low-resolution ADCs are verified as the panacea to break through those bottlenecks in cell-free MIMO system \cite{hu2019cell}, which aimed to maximize the ergodic sum-rate under limited ADC bit and investigated that large-scale antennas compensated for performance degradation by low-resolution quantization. Yuan \emph{et al.} \cite{yuan2017distributed} deduced the SE of the distributed MIMO scenario with MRC receiver and proved that SE has an upward trend with the increase of quantization bit. Nevertheless, in NAFD distributed M-MIMO, no work still exists on whether introducing low-resolution ADCs leads to severe debilitation on system SE. For other basic criteria to evaluate the low-resolution ADCs, combinatorial optimization is proposed to maximize energy efficiency (EE). More realistically, Li \emph{et al.} \cite{li2018spectral} investigated the asymptotic achievable rates with ZFR/MRC receivers in the cell-free MIMO and issued the tradeoff between SE and EE under pilot contamination. Exploring SE/EE tradeoff to obtain the feasible ADC quantization scheme in the NAFD systems is still blank in current academic research yet.

\subsection{Contribution}
Motivated by the above observations, this paper first considers the interference cancellation and beamforming training scheme in NAFD distributed M-MIMO with low-resolution ADCs. The main contributions of this paper are as follows:
\begin{itemize}
	\item With the quantization effect, UL and DL signal transmissions in NAFD distributed M-MIMO systems are modeled. The two-stage channel estimation is proposed, including pilot training and beamforming training. Moreover, SE\&EE for the NAFD distributed M-MIMO with low-resolution ADCs are discussed.
	\item The multi-objective optimization problem (MOOP) for maximizing SE and EE is proposed. The joint optimization is yielded by deep Q network (DQN) and non-dominated sorting genetic algorithm II (NSGA-II) for flexible bit allocation schemes.
	\item Monte Carlo simulation verifies the theoretical derivation with numerical results and investigates the influence of quantization noise, proving the desirability of low-resolution ADCs in the NAFD distributed M-MIMO.
\end{itemize}

\subsection{Organization and Notations}
The structure of this paper is organized as follows. Sec.~\ref{sec:system model} and Sec.~\ref{sec:ce} introduces the system model and performs channel estimation. In Sec.~\ref{sec:SE and EE}, we derive the closed-form expressions of the system SE/EE. In Sec.~\ref{sec:algorithm}, an effective ADC bit allocation algorithm is proposed. Sec.~\ref{sec:numerical result} gives the simulation results to verify theoretical derivation and draws conclusions. Sec.~\ref{sec:conclusion} summarizes the whole article. 

\emph{Notations:} Boldface lower and upper case letters denote vectors and matrices. $(\cdot)^\mathrm{H}$, $(\cdot)^\mathrm{T}$, $\mathbb{E}[\cdot]$, $\mathrm{cov}(\cdot)$ indicate the conjugate transpose, transpose, expectation and covariance operators, respectively. $\mathrm{diag}(\mathbf{A})$ returns the diagonal elements of $\mathbf{A}$, $\mathbf{x}\sim \mathcal{CN}(0,\mathbf{R})$ represents a circularly symmetric complex Gaussian vector $\mathbf{x}$ with zero mean and covariance $\mathbf{R}$. $\left| \cdot \right| $ and $\left\| \cdot\right\| $ represent the absolute value and spectral norm. $\Gamma(k,\theta)$ and $\mathrm{Nakagami}(k,\theta)$ denote the Gamma and Nakagami distribution with parameters $k$ and $\theta$. {For the abbreviations commonly used in equations, $(\cdot)_\mathrm{UL}$, $(\cdot)_\mathrm{DL}$, $(\cdot)_\mathrm{UP}$, $(\cdot)_\mathrm{DP}$, $(\cdot)_q$, $ (\cdot)^p$, $(\cdot)^\omega$, $(\cdot)_e$, $(\cdot)_\mathrm{I}$, $(\cdot)^\mathrm{pre}$ represents the UL parameter, the DL parameter, the UL pilot parameter, the DL pilot parameter, the quantized parameter, the pilot paramter, the parameter in UL/DL mode, the effective parameter, the parameter related to interference channel, and the parameter in MR/ZF scheme.}

%%%%%%%%%%%%%%%%%%%%%%%%%%%%%%%%%%%%%%%%%%%%%%%%%%%%%%%%%%%%%%%%%%%%%%%%%%%%%%%%%%%%%%

\section{System Analysis of NAFD Distributed M-MIMO}\label{sec:system model}
Fig.~\ref{system model} illustrates an NAFD M-MIMO system under low-resolution ADCs. Distributed RAUs jointly serve both UL and DL in the same time-frequency resources. UL/DL users maintain transmission simultaneously \cite{wang2019performance}. The system with low-resolution ADCs contains randomly distributed $N$ RAUs with $M$ half-duplex antennas and $K$ users with one half-duplex antenna, including $K_{\mathrm{UL}}$ UL users and $K_{\mathrm{DL}}$ DL users. {At each moment, there exist $N_{\mathrm{UL}}$ RAUs for UL reception from the UL users to UL RAUs and $N_{\mathrm{DL}}$ RAUs for DL transmission from DL RAUs to DL users, where $N_{\mathrm{UL}}+N_{\mathrm{DL}}=N$. Meanwhile, two types of interference channels, inter-user and inter-RAU between UL and DL, are considered. }
\begin{figure}[htbp]
	\centering  
	\includegraphics[width=0.4\textwidth]{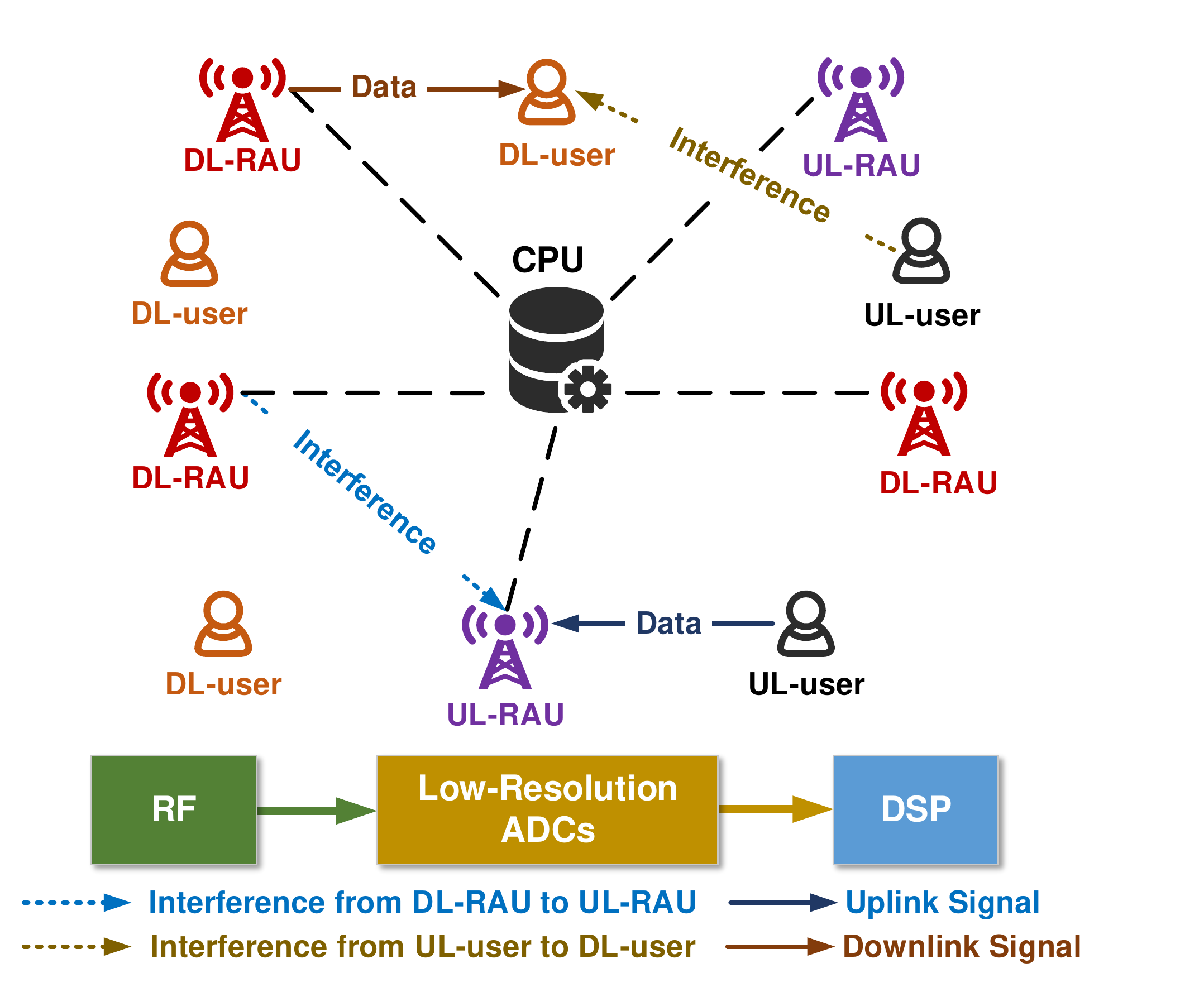}  
	\caption{System Model of NAFD Distributed M-MIMO Systems.} 
	\label{system model} 
\end{figure}
{Compared to conventional cell-free MIMO like \cite{datta2021full,anokye2020low,nguyen2020spectral,vu2019full}, NAFD's RAUs are all half-duplex, but the duplex mode of each RAU can be dynamically adjusted according to the network throughput rate, reliability or latency, enabling full duplex at the network level, which results in better SE performance than CCFD systems with suitable interference cancellation, and flexible adjustment of UL and DL RAUs to cope with transmission asymmetries effectively.}

\subsection{DL Signal Transmission with Low-Resolution ADCs}\label{subsec:dsm}
$N_{\mathrm{DL}}$ RAUs send data pre-coded by CPU in DL, and $K_{\mathrm{DL}}$ users receive signal. The data received by the $k$-th user is
\begin{equation}\label{equ:rdl}
r_{\mathrm{DL},k} = \sum_{i=1}^{K_{\mathrm{DL}}}\sqrt{p_{\mathrm{DL}}}\mu_{k,i}s_i+\sum_{j=1}^{K_{\mathrm{UL}}}\sqrt{p_{\mathrm{UL}}}u_{\mathrm{I},k,j}x_j+n_{\mathrm{DL}}.
\end{equation}
{where $p_{\mathrm{DL/UL}}$ represents the DL or UL transmit power, respectively. $s_i$ is the data signal with $\mathbb{E}[s_is_i^\mathrm{H}] = 1$ and $x_j$ is the data signal transmitted by the $j$-th UL user with
$\mathbb{E}[x_kx_k^{\mathrm{H}}] = 1$. Auxiliary channel $\mu_{k,i} = \mathbf{g}_{\mathrm{DL},k}^\mathrm{H}\mathbf{w}_i$ is obtained by the beamforming training mechanism. The pre-coding vectors $\mathbf{w}_i \in \mathbb{C}^{N_{\mathrm{DL}}M\times 1}$ for MRT and ZFT, are $\frac{\hat{\mathbf{g}}_{\mathrm{DL},i}}{\left\| \hat{\mathbf{g}}_{\mathrm{DL},i}\right\| }$ and $\frac{\hat{\mathbf{g}}_{\mathrm{DL},i}(\hat{\mathbf{g}}_{\mathrm{DL},i}^\mathrm{H}\hat{\mathbf{g}}_{\mathrm{DL},i})^{-1}}{\left\| \hat{\mathbf{g}}_{\mathrm{DL},i}(\hat{\mathbf{g}}_{\mathrm{DL},i}^\mathrm{H}\hat{\mathbf{g}}_{\mathrm{DL},i})^{-1}\right\| }$, respectively. \bl{$\mathbf{g}_{\mathrm{DL},i} = [\mathbf{g}_{\mathrm{DL},1,i}^\mathrm{H},\dots,\mathbf{g}_{\mathrm{DL},N_\mathrm{DL},i}^{\mathrm{H}}] \in \mathbb{C}^{N_{\mathrm{DL}}M\times 1}$ is the channel of the $i$-th DL user, $\hat{\mathbf{g}}_{\mathrm{DL},i}$ is the estimation of $\mathbf{g}_{\mathrm{DL},i}$, and $\mathbf{g}_{\mathrm{DL},n,i}=\lambda_{\mathrm{DL},n,i}^{1/2}\mathbf{h}_{\mathrm{DL},n,i} \in \mathbb{C}^{M\times 1}$ is the channel from the $n$-th DL RAU to the $i$-th DL user, where simplified free-space model is adopted \cite{goldsmith2005wireless}.} $\lambda_{\mathrm{DL},n,i} = d_{n,i}^{-\alpha_{\mathrm{DL}}}$ is the large-scale fading, $d_{n,i}$ is the distance between the $n$-th RAU and the $i$-th user, and $\alpha_{\mathrm{DL}}$ is the path loss. $\mathbf{h}_{\mathrm{DL},n,i}$ represents the small-scale fading modeled as Rayleigh channel \cite{li2015spectral,li2020network}. Interference channel between UL user $j$ and DL user $k$ is defined as $u_{\mathrm{I},k,j}=\lambda_{\mathrm{I},k,j}^{1/2}h_{\mathrm{I},k,j}$, and its parameters are in one-to-one correspondence with the previous definitions \cite{li2020network}. $n_{\mathrm{DL}}\sim\mathcal{CN}(0,1)$ is the DL complex additive white Gaussian noise (AWGN). The low-resolution ADCs employed on the receiving antenna quantize the analog signal. Additive quantization noise model (AQNM) \cite{fletcher2007robust} is adopted for the analysis of quantization effect, and quantizes Eq.~(\ref{equ:rdl}) as}
\begin{equation}\label{dl quantized signal}
r_{q,\mathrm{DL},k} = \xi_k r_{\mathrm{DL},k}+n_{q,\mathrm{DL},k},\\
\end{equation}
where $\xi_k=1-\rho_k$ and $n_{q,\mathrm{DL},k}$ is additive quantization Gaussian noise (AQGN) with covariance matrix $\mathbf{C}_{n_{q,\mathrm{DL},k}} = \xi_k\rho_k\mathrm{diag}(r_{\mathrm{DL},k}r_{\mathrm{DL},k}^\mathrm{H})$. $\rho_k$ is the reciprocal of the signal-to-quantization noise ratio as
\begin{table}[htp]
	\centering  
	\caption{Quantization Coefficient $\rho_k$ with $b_k$ Bits.}
	\label{table1}
	\begin{tabular}{cccccc}
		\hline
		$b$   & $1$ & $2$ &$3$ &$4$ &$\geqslant 5$\\
		\hline		
		$\rho$ & $0.3634$ &$0.1175$&$0.0345$&$0.0095$&$\sqrt{3}/2\pi\cdot 4^{-b}$\\		
		\hline
	\end{tabular} 	
\end{table}
\subsection{UL Signal Transmission with Low-Resolution ADCs}\label{subsec:usm}
$K_\mathrm{UL}$ users send the signal concurrently in UL, and $N_\mathrm{UL}$ RAUs receive jointly as
\begin{equation}\label{ulmodel}
\mathbf{r}_{\mathrm{UL}} = \sum_{i=1}^{K_{\mathrm{UL}}}\sqrt{p_{\mathrm{UL}}}\mathbf{g}_{\mathrm{UL},i}x_i+\sum_{j=1}^{K_{\mathrm{DL}}}\sqrt{p_{\mathrm{DL}}}\mathbf{f}_js_j+\mathbf{n}_{\mathrm{UL}},
\end{equation}
where $\mathbf{g}_{\mathrm{UL},i} = [\mathbf{g}_{\mathrm{UL},1,i}^{\mathrm{H}},\dots,\mathbf{g}_{\mathrm{UL},N_\mathrm{UL},i}^{\mathrm{H}}] \in \mathbb{C}^{N_{\mathrm{UL}}M\times 1}$, and $\mathbf{g}_{\mathrm{UL},n,i}=\lambda_{\mathrm{UL},n,i}^{1/2}\mathbf{h}_{\mathrm{UL},n,i} \in \mathbb{C}^{M\times 1}$ is the channel from the $n$-th UL RAU to the $i$-th UL user. $\mathbf{n}_{\mathrm{UL}}$ represents the UL complex AWGN. {$\mathbf{G}_{\mathrm{I}} \in \mathbb{C}^{N_{\mathrm{UL}}M\times N_{\mathrm{DL}}M} $ is the interference channel, and $\mathbf{g}_{\mathrm{I},i,j}=\lambda_{\mathrm{I},{i},{j}}^{1/2}\mathbf{h}_{\mathrm{I},i,j}\in\mathbb{C}^{M\times M}$ is the interference channel between $i$-th UL RAU and $j$-th DL RAU \cite{li2020network}.} Denote $\mathbf{f}_j=\mathbf{G}_{\mathrm{I}}\mathbf{w}_j \in \mathbb{C}^{N_{\mathrm{UL}}M\times 1}$ as auxiliary interference channel. Received by the UL RAU through the low-resolution ADCs, the UL quantized signal is
\begin{equation}
\mathbf{r}_{q,\mathrm{UL}} = \mathbf{A} \mathbf{r}_{\mathrm{UL}}+\mathbf{n}_{q,\mathrm{UL}},
\end{equation}
where $\mathbf{A}\!=\!\mathrm{diag}(\alpha_1,\alpha_2,\ldots,\alpha_{N_{\mathrm{UL}}})\otimes \mathbf{I}_{M} \in \mathbb{C}^{N_{\mathrm{UL}}M\times N_{\mathrm{UL}}M}$. $\alpha_n=1-\rho_n$ describes the ADC resolution of the $n$-th UL RAU. $\mathbf{n}_{q,\mathrm{UL}}$ is AQGN with $\mathbf{C}_{\mathbf{n}_{q,\mathrm{UL}}} = \mathbf{A}(\mathbf{I}-\mathbf{A})\mathrm{diag}(\mathbf{r}_{\mathrm{UL}}\mathbf{r}_{\mathrm{UL}}^\mathrm{H})$.

%%%%%%%%%%%%%%%%%%%%%%%%%%%%%%%%%%%%%%%%%%%%%%%%%%%%%%%%%%%%%%%%%%%%%%%%%%%%%%%%%%%%%%

\section{Two-Stage Channel Estimation Scheme}\label{sec:ce}
In this section, we propose a two-stage channel estimation scheme, UL pilot estimation, and beamforming training estimation. All users send pilot sequences for UL pilot estimation within the first $\tau_1$ time slots ($\tau_1 \ge K$), while all RAUs send them to the CPU for channel estimation. For beamforming training within the second $\tau_2$ time slots, all DL RAUs send pre-coded pilot, while all DL users and UL RAUs perform channel estimation, respectively. And then, the last $T-\tau_1-\tau_2$ time slots are used for data transmission. 

{In the NAFD M-MIMO system, interference between DL RAUs and UL RAUs is the main factor that reduces the UL SE when some RAUs are used for UL reception while others are used for DL transmission. To improve the SE, an interference cancellation mechanism based on beamforming training is proposed. Specifically, the CPU estimates the equivalent interference channel between the RAUs, reconstructs the interference signal using the estimated equivalent CSI, and finally performs interference cancellation by subtracting the reconstructed interference signal from the received signal.}
% Additionally, since the CPU centralizes the UL and DL baseband signals for the CLI cancellation mechanism, the CPU acquires DL pre-coded signals of all users in advance to achieve DL-to-UL interference elimination in the digital domain.

\subsection{UL Pilot Estimation}
In the pilot estimation stage, the $k$-th user sends the orthogonal pilot as $\bm{\phi}_k\in\mathbb {C}^{\tau_1\times 1}$. $n$-th RAU receives
\begin{equation}
\begin{aligned}
\mathbf{r}_{n}^{p,\omega}&=\sqrt{p_{\mathrm{UP}}}\sum_{k=1}^{K} \mathbf{g}^{\omega}_{n,k}\phi_{k}+\mathbf{n}^{p}_{n},
\end{aligned}
\end{equation}
where $p_\mathrm{UP}$ is the power of UL pilot and $\mathbf{g}^{\omega}_{n,k} \in \mathbb{C}^{M\times K} $ is the channel from $K$ user to the $n$-th RAU in UL or DL operating mode. $\mathbf{n}^{p}_{n}$ is AWGN, and its element obeys i.i.d. $\mathcal{CN}(0,\sigma_{\mathrm{UP}}^2)$. The quantized signal is 
\begin{equation}
\widetilde{\mathbf{r}}_{n}^{p,\omega}=\theta_n^\omega\left(\sqrt{p_{\mathrm{UP}}}\sum_{k=1}^{K} \mathbf{g}^{\omega}_{n,k}\phi_{k}+\mathbf{n}^p_{n}\right)+\widetilde{\mathbf{n}}^{p,\omega}_{n},
\end{equation}
where $\theta_n^\omega$ depends on the ADC quantization bits on the $n$-th RAU, and $\widetilde{\mathbf{n}}^{p,\omega}_{n}$ is AQGN. As for the $k$-th user,
\begin{equation}
\widetilde{\bm{r}}_{n,k}^{p,\omega}=\widetilde{\mathbf{r}}_{n}^{p,\omega}\phi^{\mathrm{H}}_{k}=\sqrt{p_{\mathrm{UP}}}\theta_n^\omega\mathbf{g}^{\omega}_{n,k}+\theta_n^\omega\bm{n}^p_{n,k}+\widetilde{\bm{n}}^{p,\omega}_{n,k}.
\end{equation}
Minimum mean square error (MMSE) estimation with estimated error $\widetilde{\mathbf{g}}^{\omega}_{n,k}={\mathbf{g}}^{\omega}_{n,k}-\hat{\mathbf{g}}^{\omega}_{n,k}$ is adopted by CPU \cite{hu2019cell} as
\begin{equation}\label{ghat}
\hat{\mathbf{g}}^{\omega}_{n,k}=\mathbb{E}\left[{\mathbf{g}}^{\omega}_{n,k}\right] +\frac{\mathrm{cov}({\mathbf{g}}^{\omega}_{n,k},	\tilde{\bm{r}}^{p,{\omega}}_{n,k})}{\mathrm{cov}(\tilde{\bm{r}}^{p,{\omega}}_{n,k},\tilde{\bm{r}}^{p,{\omega}}_{n,k})}\left(\tilde{\bm{r}}^{p,{\omega}}_{n,k}-\mathbb{E}\left[ 	\tilde{\bm{r}}^{p,{\omega}}_{n,k}\right] \right), 
\end{equation}
where $\widetilde{\mathbf{g}}^{\omega}_{n,k}$ and $\hat{\mathbf{g}}^{\omega}_{n,k}$ are both complex Gaussian random vectors with zero mean and variances as
\begin{subequations}
\begin{align}
\beta_{\omega,n,k}&=\frac{{p_\mathrm{UP}\theta_n^\omega}\lambda_{\omega,n,k}^2}{{p_\mathrm{UP} \lambda_{\omega,n,k}+\sigma_\mathrm{UP}^2}},\\
\eta_{\omega,n,k}&=\frac{(1-\theta_n^\omega)p_\mathrm{UP}\lambda_{\omega,n,k}+\lambda_{\omega,n,k}\sigma_{\mathrm{UP}}^2}{p_\mathrm{UP}\lambda_{\omega,n,k}+\sigma_\mathrm{UP}^2}.
\end{align}
\end{subequations}
Due to limited space, the details can be referred to \cite{hu2019cell}.

\subsection{Beamforming Training Estimation}
In the beamforming training stage, DL RAUs send pilot. DL users and UL RAUs estimate the equivalent channels $\mu_{k,i},\mathbf{f}_j$. 
\subsubsection{DL User Estimation} 
RAUs send $\mathbf{W}\mathbf{P}\bm{\Psi}$, where $\bm{\Psi}\in \mathbb{C}^{K_{\mathrm{DL}} \times\tau_2} $ is orthogonal pilot sequence. Define $\mathbf{W}=[\mathbf{w}_1,...,\mathbf{w}_{K_\mathrm{DL}}] \in \mathbb{C}^{{N_{\mathrm{DL}}M \times K_{\mathrm{DL}} }} $ as the beamforming matrix and $\mathbf{P}=\mathrm{diag}(\sqrt{p_\mathrm{DP,1}},\ldots,\sqrt{p_{\mathrm{DP},K_{\mathrm{DL}}}})$ as the pilot power matrix, where $p_\mathrm{DP}$ is the power of DL pilot. The received signal of the $k$-th DL user is
\begin{equation}
\mathbf{r}_{k}^{p}=\mathbf{g}_{\mathrm{DL},k}^\mathrm{H}\mathbf{W}\mathbf{P}\bm{\Psi}+\mathbf{n}_\mathrm{DP},
\end{equation}
where $\mathbf{n}_\mathrm{DP}$ is AWGN, and its element obeys i.i.d. $\mathcal{CN}(0,\sigma_{\mathrm{DP}}^2)$. The quantized signal is
\begin{equation}
\tilde{\mathbf{r}}_{k}^{p}=\xi_k\widetilde{\mathbf{P}}\bm{\Psi}+\xi_k\mathbf{n}_\mathrm{DP}+\tilde{\mathbf{n}}_\mathrm{DP},
\end{equation}
where $\widetilde{\mathbf{P}}=[\sqrt{p_\mathrm{DP,1}}\mu_{k,1},\ldots,\sqrt{p_{\mathrm{DP},K_{\mathrm{DL}}}}\mu_{k,K_{\mathrm{DL}}}]$. Correlating the quantized signal with $\bm{\Psi}^\mathrm{H}$, we have the effective signal as
\begin{equation}\label{effective signal}
\tilde{\mathbf{r}}_{e,\mathrm{k}}^{p}=\tilde{\mathbf{r}}_{k}^{p}\bm{\Psi}^\mathrm{H}=\xi_k\widetilde{\mathbf{P}}+\xi_k\mathbf{n}_{e,\mathrm{DP}}+\tilde{\mathbf{n}}_{e,\mathrm{DP}},
\end{equation}
where $\mathbf{n}_{e,\mathrm{DP}}=\mathbf{n}_\mathrm{DP}\bm{\Psi}^\mathrm{H}$, $\tilde{\mathbf{n}}_{e,\mathrm{DP}}=\tilde{\mathbf{n}}_{e,\mathrm{DP}}\bm{\Psi}^\mathrm{H}$. In addition, we obtain the estimation of $\mu_{k,i}$ by MMSE algorithm as
\begin{equation}\label{equ:muhat}
\hat{\mu}_{k,i}=\mathbb{E}\left[ \mu_{k,i}\right] +\frac{\mathrm{cov}(\mu_{k,i},	\tilde{\bm{r}}_{e,k}^{p,i})}{\mathrm{cov}(\tilde{\bm{r}}_{e,k}^{p,i},	\tilde{\bm{r}}_{e,k}^{p,i})}\left(\tilde{\bm{r}}_{e,k}^{p,i}-\mathbb{E}\left[ 	\tilde{\bm{r}}_{e,k}^{p,i}\right] \right), 
\end{equation}
where $\tilde{\bm{r}}_{e,k}^{p,i}$ is the $i$-th column of $\tilde{\mathbf{r}}_{e,k}^{p}$. With Lemma~\ref{Lemma 1} in Appendix~\ref{appendix 1}, $
||\hat{\mathbf{g}}_{\mathrm{DL},k}||^2\sim \Gamma(k_{\mathrm{DL},k},\theta_{\mathrm{DL},k}),||\widetilde{\mathbf{g}}_{\mathrm{DL},k}||^2\sim \Gamma(\tilde{k}_{\mathrm{DL},k},\tilde{\theta}_{\mathrm{DL},k})$. The details of the distribution parameters are provided in Appendix~\ref{appendix 2}.
% and the closed-form expressions for Eq.~(\ref{equ:muhat}) with MRT and ZF are given in the following propositions.
\begin{Prop} \label{Lemma2}
With MRT/ZFT pre-coders ($\mathrm{pre}\in\{\mathrm{MR,ZF}\}$) on DL RAUs, MMSE estimation can be expressed as
\begin{small}
\begin{equation}
\hat{\mu}_{k,i}^{\mathrm{pre}}=
\left\{
\begin{aligned}
&E_k^\mathrm{pre}+\frac{\sqrt{p_{\mathrm{DP}}}\bar{\chi}_k^\mathrm{pre}\tilde{\bm{r}}_{e,k}^{p,i}-\xi_k\bar{\chi}_k^\mathrm{pre}{p_{\mathrm{DP}}}E_k^\mathrm{pre}}{p_{\mathrm{DP}}\widetilde{\chi}^{\mathrm{pre}}_k+\sigma_{\mathrm{DP}}^2},i \neq k \\
&\frac{\sqrt{p_{\mathrm{DP}}}{\chi}_k^\mathrm{pre}}{{p_{\mathrm{DP}}}{\chi}_k^\mathrm{pre}+\sigma_{\mathrm{DP}}^2}, i=k
\end{aligned}
\right..
\end{equation}
\end{small}
\begin{itemize}
	\item For MRT pre-coding, 
	\begin{small}
	\begin{subequations}
	\begin{align}
	E_k^\mathrm{MR}&=\frac{\Gamma\left( k_{\mathrm{DL},k}+{1}/{2}\right) }{\Gamma\left( k_{\mathrm{DL},k}\right)}\theta_{\mathrm{DL},k}^{1/2},\\
	\widetilde{\chi}^{\mathrm{MR}}_k&=\bar{\chi}_k^\mathrm{MR}+(1-\xi_k)(E_k^\mathrm{MR})^2\\
	\bar{\chi}_k^\mathrm{MR}&=k_{\mathrm{DL},k}\theta_{\mathrm{DL},k}-\left( E_k^\mathrm{MR}\right) ^2+\frac{\tilde{k}_{\mathrm{DL},k}\tilde{\theta}_{\mathrm{DL},k}}{N_{\mathrm{DL}}M},\\
	\chi_k^\mathrm{MR}&=\frac{ k_{\mathrm{DL},k}\theta_{\mathrm{DL},k}+\tilde{k}_{\mathrm{DL},k}\tilde{\theta}_{\mathrm{DL},k}}{N_{\mathrm{DL}}M},
	\end{align}
	\end{subequations}
	\end{small}

	\item For ZFT pre-coding,
	\begin{small}
	\begin{subequations}
	\begin{align}
	E_k^\mathrm{ZF}&=\frac{\Gamma\left( t_{\mathrm{DL}}k_{\mathrm{DL},k}+{1}/{2}\right) }{\Gamma\left( t_{\mathrm{DL}}k_{\mathrm{DL},k}\right)}\theta_{\mathrm{DL},k}^{1/2},\\
	\widetilde{\chi}^{\mathrm{ZF}}_k&=\bar{\chi}_k^\mathrm{ZF}+(1-\xi_k)(E_k^\mathrm{ZF})^2\\
	\bar{\chi}_k^\mathrm{ZF}&=t_{\mathrm{DL}}k_{\mathrm{DL},k}\theta_{\mathrm{DL},k}-\left( E_k^\mathrm{ZF}\right)^2+\frac{\tilde{k}_{\mathrm{DL},k}\tilde{\theta}_{\mathrm{DL},k}}{N_{\mathrm{DL}}M},\\
	\chi_k^\mathrm{ZF}&=\frac{k_{\mathrm{DL},k}\theta_{\mathrm{DL},k}+\tilde{k}_{\mathrm{DL},k}\tilde{\theta}_{\mathrm{DL},k}}{N_{\mathrm{DL}}M},\\
	t_{\mathrm{DL/UL}}&=\frac{N_{\mathrm{DL/UL}}M-K_{\mathrm{DL/UL}}+1}{N_{\mathrm{DL/UL}}M}.
	\end{align}
	\end{subequations}
	\end{small}
\end{itemize}
\end{Prop}
\textit{Proof:} Please refer to Appendix \ref{appendix 2}.
\begin{Rem}
With a reduced channel-hardening effect, the statistical CSI, without a priori knowledge of the perfect CSI, leads to severe performance degradation. Beamforming training scheme is adopted to estimate the equivalent channel $\mu_{k,i}$ instead of the channel matrix $\mathbf{g}_{k,i}$, effectively reducing the number of pilot sequences from $N_{\mathrm{DL}}M$ to $K_{\mathrm{DL}}$.
\end{Rem}

\subsubsection{UL RAUs Estimation}
For the estimation of interference channel $\mathbf{G}_\mathrm{I}$ to fulfill CLI cancellation, {all UL RAUs jointly receive the pilot from the DL RAUs and send it back to CPU via backhaul link as}
\begin{equation}
\begin{aligned}
\mathbf{Y}^p&=\mathbf{G}_\mathrm{I}\mathbf{W}\mathbf{P}\bm{\Psi}+\mathbf{N}_\mathrm{UP}\\
&=\left[ \sqrt{p_\mathrm{DP,1}}\mathbf{f}_1,\ldots,\sqrt{p_{\mathrm{DP},K_{\mathrm{DL}}}}\mathbf{f}_{K_{\mathrm{DL}}}\right]\bm{\Psi}+\mathbf{N}_\mathrm{UP},
\end{aligned}
\end{equation}
where DL RAUs send $\mathbf{W}\mathbf{P}\bm{\Psi}$, similar as above. $\mathbf{N}_\mathrm{UP}$ is AWGN. Correlating the received signal with $\bm{\Psi}^\mathrm{H}$ and extracting the $i$-th column, the quantized signal is 
\begin{equation}\label{y_i}
\tilde{\mathbf{y}}_i^p=\mathbf{A}\sqrt{p_{\mathrm{DP}}}\mathbf{f}_i+\mathbf{A}\mathbf{n}_{\mathrm{UP},i}+\tilde{\mathbf{n}}_{\mathrm{UP},i},
\end{equation}
{where $\tilde{\mathbf{n}}_{\mathrm{UP},i}$ is AWGN with covariance $\mathbf{C}_{\tilde{\mathbf{n}}_{\mathrm{UP},i}} = \mathbf{A}(\mathbf{I}-\mathbf{A})\mathrm{diag}(||\tilde{\mathbf{y}}_i^p||^2)$.} The MMSE estimation of $\mathbf{f}_i$ is
\begin{equation}\label{fihat}
\hat{\mathbf{f}_i}=\mathbb{E}\left[ \mathbf{f}_i\right] +\frac{\mathrm{cov}(\mathbf{f}_i,	\tilde{\mathbf{y}}_i^p)}{\mathrm{cov}(\tilde{\mathbf{y}}_i^p,	\tilde{\mathbf{y}}_i^p)}\left(\tilde{\mathbf{y}}_i^p-\mathbb{E}\left[ 	\tilde{\mathbf{y}}_i^p\right] \right).
\end{equation}

\begin{Prop} \label{Lemma3}
In the NAFD distributed M-MIMO systems, the estimations of the effective interference channel is
\begin{equation}\label{fhat}
\hat{\mathbf{f}_i}=\sqrt{p_{\mathrm{DP}}}\bm{\Lambda}\left( p_{\mathrm{DP}}\bm{\Lambda}+\sigma_{\mathrm{UP}}^2\mathbf{I}_{MN_{\mathrm{UL}}}\right) \tilde{\mathbf{y}}_i^p,
\end{equation}
where $\bm{\Lambda}=\mathrm{diag}\left( \sum_{n=1}^{N_{\mathrm{DL}}}\lambda_{\mathrm{I},1,n},\ldots,\sum_{n=1}^{N_{\mathrm{DL}}}\lambda_{{\mathrm{I},N_{\mathrm{UL}},n}}\right)\otimes\frac{\mathbf{I}_M}{N_{\mathrm{DL}}}$. \bl{Decomposing} it into the multiplication of equivalent large-scale fading $\delta_{\mathrm{I},n,i}$ and small-scale fading $\hat{\bm{q}}_{\mathrm{I},n,i}$, $\hat{\mathbf{f}_i}$ is rewritten as
\begin{equation}\label{delta}
\hat{\mathbf{f}_i}=[
\delta_{\mathrm{I},1,i}\hat{\bm{q}}_{\mathrm{I},1,i},
\delta_{\mathrm{I},2,i}\hat{\bm{q}}_{\mathrm{I},2,i},
\ldots,
\delta_{\mathrm{I},N_{\mathrm{UL}},i}\hat{\bm{q}}_{\mathrm{I},N_{\mathrm{UL}},i}
]^\mathrm{H},
\end{equation}	
where
\begin{equation}
\delta_{\mathrm{I},j,i}=\frac{\sqrt{p_{\mathrm{DP}}\alpha_j}\sum_{n=1}^{N_{\mathrm{DL}}}\lambda_{\mathrm{I},j,n}}{\sqrt{N_{\mathrm{DL}}p_{\mathrm{DP}}\sum_{n=1}^{N_{\mathrm{DL}}}\lambda_{\mathrm{I},j,n}+N_{\mathrm{DL}}^2\sigma_{\mathrm{UP}}^2}}.
\end{equation}
\end{Prop}

\textit{Proof:} Please refer to Appendix \ref{appendix 3}.
% \begin{Rem}
% {Since the effective channel interference matrix is not correlated with the beamforming matrix, the same effective channel interference matrix is obtained for different beamforming schemes, which can be used for interference cancellation between the UL and DL RAUs.}
% \end{Rem}

%%%%%%%%%%%%%%%%%%%%%%%%%%%%%%%%%%%%%%%%%%%%%%%%%%%%%%%%%%%%%%%%%%%%%%%%%%%%%%%%%%%%%%%%%

\section{SE/EE Analysis of NAFD Distributed M-MIMO}\label{sec:SE and EE}
\subsection{Derivation of the DL Achievable Rate}
With the estimated CSI and statistical CSI of DL users, we analyze the DL SE of the NAFD distributed M-MIMO systems with low-resolution ADCs. With the DL beamforming training estimation, we rewrite the received signal of the $k$-th user as
\begin{equation}\label{dl signal}
r^q_{\mathrm{DL},k} =\underbrace{\xi_k\sqrt{p_{\mathrm{DL}}}\hat{\mu}_{k,k}s_k}_\text{desird signal}+ \underbrace{n^{\text{eff}}_k+ \xi_kn_{\mathrm{DL}}+n^q_{\mathrm{DL},k}}_\text{effective noise}
\end{equation}
{where $n^{\text{eff}}_k=\xi_k\sum_{i\neq k}\sqrt{p_{\mathrm{DL}}}\hat{\mu}_{k,i}s_i+\xi_k\sum_{i=1}^{K_{\mathrm{DL}}}\sqrt{p_{\mathrm{DL}}}\tilde{\mu}_{k,i}s_i+\xi_k\sum_{j=1}^{K_{\mathrm{UL}}}\sqrt{p_{\mathrm{UL}}}u_{t,k,j}x_j$, and $\tilde{\mu}_{k,i}$ represents the estimation error of $\mu_{k,i}$. For simplicity, each term in {effective noise} can be approximated as equivalent noise obeying the Gaussian distribution. The DL achievable rate of the $k$-th user is derived as Eq.~(\ref{Rdl}), where $\widetilde{\mathbf{C}}_{n_{\mathrm{DL},k}}=\sigma_{\mathrm{DL}}^2+\mathbf{C}_{n_{q,\mathrm{DL},k}}/\xi_k^2$.}

\begin{figure*}[htpb]
\begin{small}
\begin{equation}\label{Rdl}
R_{\mathrm{DL},k}=\mathbb{E}\left[ \mathrm{log}_2\left( 1+\frac{\left| \sqrt{p_{\mathrm{DL}}}\hat{\mu}_{k,k}s_k\right|^2 }{\left| \sum_{i\neq k}\sqrt{p_{\mathrm{DL}}}\hat{\mu}_{k,i}s_i\right|^2+\left| \sum_{i=1}^{K_{\mathrm{DL}}}\sqrt{p_{\mathrm{DL}}}\tilde{\mu}_{k,i}s_i\right|^2+\left| \sum_{j=1}^{K_{\mathrm{UL}}}\sqrt{p_{\mathrm{UL}}}u_{\mathrm{I},k,j}x_j\right|^2+\widetilde{\mathbf{C}}_{n_{\mathrm{DL},k}} }\right) \right] .
\end{equation}
\begin{equation}\label{Rul}
R_{\mathrm{UL},k}=\mathbb{E}\left[ \mathrm{log}_2\left( 1+\frac{p_{\mathrm{UL}}\left| \mathbf{v}_k^\mathrm{H}\mathbf{A}\hat{\mathbf{g}}_{\mathrm{UL},k}\right|^2 }{\sum_{i\neq k}p_{\mathrm{UL}}\left| \mathbf{v}_k^\mathrm{H}\mathbf{A}\hat{\mathbf{g}}_{\mathrm{UL},i}\right|^2+\sum_{i=1}^{K_{\mathrm{UL}}}p_{\mathrm{UL}}\left| \mathbf{v}_k^\mathrm{H}\mathbf{A}\tilde{\mathbf{g}}_{\mathrm{UL},i}\right|^2+\sum_{j=1}^{K_{\mathrm{DL}}}p_{\mathrm{DL}}\left| \mathbf{v}_k^\mathrm{H}\mathbf{A}\tilde{\mathbf{f}}_{j}\right|^2 +\mathbf{\widetilde{C}}_{\mathbf{n}_{\mathrm{UL},k}}}\right) \right]. \tag{33}
\end{equation}
\end{small}
\hrulefill
\end{figure*}

\begin{thm}\label{thm1}
	Under the estimated CSI at the DL user, the closed-form expression of the DL achievable rate of $k$-th user is
	% \begin{small}
	\begin{equation}\label{Rdlpre}
	R_{\mathrm{DL},k}^\mathrm{pre}=\mathrm{log}_2\left( 1+\frac{\xi_kp_{\mathrm{DL}}\left[(E_k^\mathrm{pre})^2+F_k^\mathrm{pre}\right]}{A_k^\mathrm{pre} + B_k^\mathrm{pre}+C_k^\mathrm{pre}+\sigma_{\mathrm{DL}}^2 }\right),
	\end{equation}
	% \end{small}
	where
	\begin{subequations}\label{dlparts}
	\begin{align}
	&F_k^\mathrm{pre}= \frac{p_{\mathrm{DP}}(\bar{\chi}_k^\mathrm{pre})^2}{p_{\mathrm{DP}}[\bar{\chi}_k^\mathrm{pre}+(1-\xi_k)(E_k^\mathrm{pre})^2]+\sigma_{\mathrm{DP}}^2}, \\
	&A_k^\mathrm{pre}=\sum_{i\neq k}\xi_kp_{\mathrm{DL}}\chi_k^\mathrm{pre}+\sum_{j=1}^{K_{\mathrm{UL}}}p_{\mathrm{UL}}\lambda_{\mathrm{I},k,j},\\
	&B_k^\mathrm{pre}=\frac{\xi_k(1-\xi_k)p_{\mathrm{DP}}\bar{\chi}_k^\mathrm{pre}(E_k^\mathrm{pre})^2+\xi_k\bar{\chi}_k^\mathrm{pre}\sigma_{\mathrm{\mathrm{DP}}}^2}{\bar{\chi}_k^\mathrm{pre}+(1-\xi_k)(E_k^\mathrm{pre})^2+\sigma_{\mathrm{DP}}^2/p_{\mathrm{DL}}},\\
	&C_k^\mathrm{pre}=(1-\xi_k)p_{\mathrm{DL}}\left[ \bar{\chi}_k^\mathrm{pre}+(E_k^\mathrm{pre})^2\right]. 
	\end{align}
	\end{subequations}
\end{thm} 

\textit{Proof:} Please refer to Appendix \ref{appendix 4}.

\begin{thm}\label{thm2}
	Under the statistical CSI, the closed-form expression of the DL achievable rate of $k$-th user is
	\begin{equation}\label{Rdlpre2}
	R_{\mathrm{DL},k,s}^\mathrm{pre}=\mathrm{log}_2\left( 1+\frac{\xi_kp_{\mathrm{DL}} (E_k^\mathrm{pre})^2}{A_{k,s}^\mathrm{pre} + B_{k,s}^\mathrm{pre}+\sigma_{\mathrm{DL}}^2 }\right) ,
	\end{equation}
	where
	\begin{subequations}\label{dlparts2}
	\begin{align}
	&A^\mathrm{pre}_{k,s}=\sum_{i\neq k}\xi_kp_{\mathrm{DL}}\chi_k^\mathrm{pre}+p_{\mathrm{DL}}\bar{\chi}_k^\mathrm{pre},\\
	&B^\mathrm{pre}_{k,s}=\sum_{j=1}^{K_{\mathrm{UL}}}p_{\mathrm{UL}}\lambda_{\mathrm{I},k,j}+(1-\xi_k)p_{\mathrm{DL}} (E_k^\mathrm{pre})^2.
	\end{align}
	\end{subequations}
\end{thm} 
\begin{proof}
	The only difference between Theorems~\ref{thm2} and \ref{thm1} exists in the gap between the effective signal:
	\begin{equation}\label{equ:gai}
	\begin{aligned}
	&\mathbb{E}\left[ \left| \xi_k\sqrt{p_{\mathrm{DL}}}\hat{\mu}_{k,k}s_k\right|^2\right]-\left( \mathbb{E} \left[ \xi_k\sqrt{p_{\mathrm{DL}}}{\mu}_{k,k}s_k\right]\right) ^2\\
	&=\xi_k^2p_{\mathrm{DL}}\frac{p_{\mathrm{DP}}(\bar{\chi}_k^\mathrm{pre})^2}{p_{\mathrm{DP}}\bar{\chi}_k^\mathrm{pre}+(1-\xi_k)p_{\mathrm{DP}}(E_k^\mathrm{pre})^2+\sigma_{\mathrm{DP}}^2}.
	\end{aligned}
	\end{equation}
	After inserting Eq.~(\ref{equ:gai}) into denominator, we have Eqs.~(\ref{Rdlpre2}) and (\ref{dlparts2}) with proper approximation.
\end{proof}
\begin{Rem}
{The closed-form expressions of the DL achievable rate obtained by the DL channel estimation and static statistics are given, respectively. The beamforming training mechanism's effect can be analyzed by comparing the two cases.}
\end{Rem}
\subsection{Derivation of the UL Achievable Rate}
Considering interference cancellation at the UL RAUs, we analyze the UL SE of NAFD distributed M-MIMO systems with low-resolution ADCs. The quantized UL signal is 
% \begin{small}
\begin{equation}
\mathbf{r}_{q}^{\mathrm{UL}}\!=\!\mathbf{A}\left(\sum_{i=1}^{K_{\mathrm{UL}}}\sqrt{p_{\mathrm{UL}}}\mathbf{g}_{\mathrm{UL},i}x_i\!+\!\sum_{j=1}^{K_{\mathrm{DL}}}\sqrt{p_{\mathrm{DL}}}{\mathbf{f}}_js_j\right)\!+\!\widetilde{\mathbf{n}}_{q,\mathrm{UL}},
\end{equation}
% \end{small}
where $\widetilde{\mathbf{n}}_{q,\mathrm{UL}}=\mathbf{A}\mathbf{n}_{\mathrm{UL}}+\mathbf{n}_{q,\mathrm{UL}}$ is the equivalent quantization noise. The UL RAUs perform CLI cancellation as
% \begin{small}
\begin{equation}
\begin{aligned}
&\mathbf{r}_{q,\mathrm{UL}}^\mathrm{IC}
=\mathbf{r}_{q}^{\mathrm{UL}}-\mathbf{x}^\mathrm{IC},\\
% &=\mathbf{A}\left(\sum_{i=1}^{K_{\mathrm{UL}}}\sqrt{p_{\mathrm{UL},i}}\mathbf{g}_{\mathrm{UL},i}x_i+\sum_{j=1}^{K_{\mathrm{DL}}}\sqrt{p_{\mathrm{DL}}}\tilde{\mathbf{f}}_js_j+\mathbf{n}_{\mathrm{UL}}\right)+\mathbf{n}_{\mathrm{q,ul}},
\end{aligned}
\end{equation}
% \end{small}
where $\mathbf{x}^\mathrm{IC}=\sum_{j=1}^{K_{\mathrm{DL}}}\sqrt{p_{\mathrm{DL} ,j}}\mathbf{A}\hat{\mathbf{f}}_js_j \in \mathbb{C}^{N_{\mathrm{UL}}M\times 1}$ is eliminable interference and $\tilde{\mathbf{f}}_j={\mathbf{f}}_j-\hat{\mathbf{f}}_j$ is the estimation error of ${\mathbf{f}}_j$. After interference cancellation, the signal is processed by the receiver as
% \begin{small}
\begin{equation}\label{equ:dcp}
\begin{aligned}
&{r}_{k}^{\mathrm{UL}}=\mathbf{v}_k^\mathrm{H}\mathbf{r}_{q,\mathrm{UL}}^\mathrm{IC}=\mathbf{v}_k^\mathrm{H}{\mathbf{n}}_{q,\mathrm{UL}}+{\mathbf{n}}_{\mathrm{UL}}\\
&+\mathbf{v}_k^\mathrm{H}\mathbf{A}\left(\sum_{i=1}^{K_{\mathrm{UL}}}\sqrt{p_{\mathrm{UL}}}\mathbf{g}_{\mathrm{UL},i}x_i+\sum_{j=1}^{K_{\mathrm{DL}}}\sqrt{p_{\mathrm{DL}}}\tilde{\mathbf{f}}_js_j\right), 
\end{aligned}
\end{equation}
% \end{small}
where the receiving vector $\mathbf{v}_k^\mathrm{H}$ for the $k$-th user for MRC and ZFR are ${\hat{\mathbf{g}}_{\mathrm{UL},k}}$ and $\hat{\mathbf{g}}_{\mathrm{UL},k}\left( \hat{\mathbf{g}}_{\mathrm{UL},k}^\mathrm{H}\hat{\mathbf{g}}_{\mathrm{UL},k}\right)^{-1}$, respectively. After decomposing ${r}_{k}^{\mathrm{UL}}$ into effective signal, interference and equivalent noise in Eq.~(\ref{equ:dcp}), the UL achievable rate of the $k$-th user can be expressed as Eq.~(\ref{Rul}), where $\mathbf{\widetilde{C}}_{\mathbf{n}_{\mathrm{UL},k}}=\left\|\mathbf{v}_k^\mathrm{H}\mathbf{A} \right\|^2\sigma_{\mathrm{UL}}^2+\mathbf{v}_k^\mathrm{H}\mathbf{C}_{\mathbf{n}_{q,\mathrm{UL}}}\mathbf{v}_k$. 
\begin{thm}\label{thm3}
	With the interference cancellation of UL RAUs, the closed-form expression of the UL achievable rate of the NAFD distributed M-MIMO systems can be derived as
	\begin{equation}\label{RulMRC}
	\setcounter{equation}{34}
	R_{\mathrm{UL},k}^\mathrm{pre}=\mathrm{log}_2\left( 1+\frac{A_{k}^\mathrm{pre}}{B_{k}^\mathrm{pre}+\sum_{n=1}^{N_\mathrm{UL}}\alpha_n C_{n,k}^\mathrm{pre}/N_{\mathrm{UL}}}\right), 
	\end{equation}

	\begin{itemize}		
		\item With MRT pre-coding and MRC receiver,
				\begin{small}
		\begin{subequations}
		\begin{align}
		&A_{k}^\mathrm{MR}=p_{\mathrm{UL}}M\sum_{n=1}^{N_{\mathrm{UL}}}\alpha_n^2\beta_{\mathrm{UL},n,k}+{\sigma_{\mathrm{UL}}^2}/{N_{\mathrm{UL}}}\sum_{n=1}^{N_{\mathrm{UL}}}\alpha_n,\\
		&B_{k}^\mathrm{MR}=p_{\mathrm{UL}}M\sum_{n=1}^{N_{\mathrm{UL}}}\alpha_n\left( 1-\alpha_n\right) \beta_{\mathrm{UL},n,k},\\
		&C_{n,k}^\mathrm{MR}=p_{\mathrm{DL}}\overbrace{\left(\sum_{m=1}^{N_{\mathrm{DL}}}\lambda_{\mathrm{I},n,m}-\delta_{\mathrm{UL},n,i}^2\right)/{N_{\mathrm{DL}}}}^{\rho_{\mathrm{UL},n,i}}\\\nonumber
		&+\sum_{i\neq k}p_{\mathrm{UL}}\beta_{\mathrm{UL},n,i}+\sum_{i=1}^{K_{\mathrm{UL}}}p_{\mathrm{UL}}\eta_{\mathrm{UL},n,i},
		\end{align}
		\end{subequations}
				\end{small}
		\item With ZFT pre-coding and ZFR receiver, 
		\begin{small}
		\begin{subequations}
		\begin{align}
		&A_{k}^\mathrm{ZF}=t_{\mathrm{UL}}p_{\mathrm{UL}}M\sum_{n=1}^{N_{\mathrm{UL}}}\alpha_n^2\beta_{\mathrm{UL},n,k},\\
		&B_{k}^\mathrm{ZF}=t_{\mathrm{UL}}p_{\mathrm{UL}}M\sum_{n=1}^{N_{\mathrm{UL}}}\alpha_n\left( 1-\alpha_n\right) \beta_{\mathrm{UL},n,k},\\
		&C_{n,k}^\mathrm{ZF}=\sum_{j=1}^{K_{\mathrm{DL}}}p_{\mathrm{DL}}\rho_{\mathrm{UL},n,j}^2+\sum_{i=1}^{K_{\mathrm{UL}}}p_{\mathrm{UL}}\eta_{\mathrm{UL},n,i}+\sigma_{\mathrm{UL}}^2.
		\end{align}
		\end{subequations}
		\end{small}
	\end{itemize}
\end{thm} 

\textit{Proof:} Please refer to Appendix \ref{appendix 5}.
\begin{Rem}
With the CLI cancellation under low-resolution ADCs, we obtain the expression for closed-form expressions of the UL SE of the NAFD distributed M-MIMO. The total energy consumption will be analyzed to further analyze the impact of low-resolution ADCs comprehensively.
\end{Rem}

\subsection{Power Consumption of NAFD Distributed M-MIMO}
The high-resolution ADCs with high SE pay higher hardware cost and energy consumption. Environment-friendly communication raises an urgent call for the energy consumption with low-resolution ADCs. {Inspired by \cite{lv2018spectral,zuo2016energy}, the EE of the NAFD distributed M-MIMO systems is modeled as
\begin{equation}
\eta_\mathrm{EE}=\frac{W\cdot\left( \sum_{k=1}^{K_\mathrm{UL}}R_{\mathrm{UL},k}+\sum_{k=1}^{K_\mathrm{DL}}R_{\mathrm{DL},k}\right)}{P_\mathrm{TC}+P_\mathrm{T}+P_\mathrm{LP}+P_\mathrm{BH}},
\end{equation}
where $W$ is the transmission bandwidth. $P_\mathrm{TC}$ is the power consumption of the transceiver link\footnote{{In this paper, the effect of the low-resolution ADC is mainly considered to evaluate
the contribution for EE. Other minor factors are modelled uniformly in $P_\text{UE}$ (such as the
power amplifier in the transmitter chain and the low-noise-amplifier in the receiver chain) as \cite{zuo2016energy,anokye2020low,li2015spectral}.}} as

\begin{equation}
P_\mathrm{TC}=NMP_\mathrm{RAU}+\rho_\mathrm{SYN} P_\mathrm{SYN}+KP_\mathrm{UE}+MP_\mathrm{ADC}.
\end{equation}
where $P_\mathrm{RAU}, P_\mathrm{SYN}$ and $P_\mathrm{UE}$ represent the power consumption of the RAU, local oscillator, and the user, respectively. In the distributed antenna systems, each RAU has a crystal oscillator, $\rho_\mathrm{SYN}=N$. In the co-located antenna systems, base station only deploys one crystal oscillator, $\rho_\mathrm{SYN}=1$. $P_\mathrm{ADC}$ is the power consumption of the low-resolution ADCs \cite{lv2018spectral}, as $P_\mathrm{ADC}=\sum_i a_0\cdot M\cdot 2^{b_i}+a_1$, where $a_0$ and $a_1$ are constant parameters related to energy consumption of ADCs \cite{lv2018spectral} and $b_i$ is the bitwidth of $i$-th ADC.} $P_\mathrm{T}$ is the energy consumption of the transmitting signals as

\begin{equation}
P_\mathrm{T}=\frac{K_\mathrm{UL}(T-\tau_1-\tau_2)}{T\cdot \xi}p_{\mathrm{UL}}+\frac{K_\mathrm{DL}(T-\tau_1-\tau_2)}{T\cdot \xi}p_{\mathrm{DL}},
\end{equation}
where $\xi$ is the amplifier efficiency. Denote $P_\mathrm{LP}$ as the energy consumption of the linear receiver and pre-coder \cite{lv2018spectral}: 

\begin{equation}
P_\mathrm{LP}\!=\!\frac{T\!-\!\tau_{\omega}}{T}\frac{2WMNK}{L_\mathrm{RAU}}\!+\!\frac{\tau_{\omega}}{T}
\left\{
\begin{aligned}
&\frac{WMNK(3K\!+\!1)}{L_\mathrm{RAU}},(a)\\
&\frac{3WMNK}{L_\mathrm{RAU}},(b)
\end{aligned}
\right.
\end{equation}
where $(a)$ and $(b)$ mean RAUs adopt MRT pre-coder and MRC receiver, and ZFT pre-coder and ZFR receiver, respectively. $L_\mathrm{RAU}$ is the calculation efficiency of the complex operation under unit power consumption. {$P_\mathrm{BH}$, which can be ignored in the co-located antenna systems, is the energy consumption of the backhaul link \cite{lv2018spectral} between RAUs and CPU as 
\begin{equation}
P_\mathrm{BH}=N\left[P_\mathrm{0}+WP_\mathrm{BT}\left( \sum_{k=1}^{K_\mathrm{UL}}R_{\mathrm{UL},k}+\sum_{k=1}^{K_\mathrm{DL}}R_{\mathrm{DL},k}\right)\right],
\end{equation}
where $P_\mathrm{0}$ and $P_\mathrm{BT}$ are the fixed energy consumption on each backhaul link and dynamic energy consumption associated with the backhaul rate, respectively.}

\begin{Rem}
We discuss the various components of energy consumption in NAFD distributed M-MIMO systems and EE with low-resolution ADCs. The energy consumption of the transceiver can be significantly reduced by a flexible selection of quantization bits, which improves the system EE. However, the reduction of quantization bits affects the SE of the system at the same time, and in the next section, we propose the corresponding allocation scheme for feasible ADC quantization.
\end{Rem}

%%%%%%%%%%%%%%%%%%%%%%%%%%%%%%%%%%%%%%%%%%%%%%%%%%%%%%%%%%%%%%%%%%%%%%%%%%%%%%%%%%%%%%%%%
%%%%%%%%%%%%%%%%%%%%%%%%%%%%%%%%%%%%%%%%%%%%%%%%%%%%%%%%%%%%%%%%%%%%%%%%%%%%%%%%%%%%%%%%%
%%%%%%%%%%%%%%%%%%%%%%%%%%%%%%%%%%%%%%%%%%%%%%%%%%%%%%%%%%%%%%%%%%%%%%%%%%%%%%%%%%%%%%%%%

\section{Bit Allocation Algorithm for ADC Quantization}\label{sec:algorithm}
In this section, we discuss the delicate allocation mechanism of ADC quantization for RAUs and users in NAFD distributed M-MIMO systems, which considers the tradeoff between SE and EE to significantly reduce the system energy consumption while ensuring the system SE.

\subsection{Problem Description of ADC Bit Allocation}
Adaptive numerical schemes with a bit allocation mechanism show the non-negligible advantage of greater flexibility, and greener power consumption than the equal-resolution schemes \cite{li2018spectral}. The objective function is designed to reach the maximum value of total SE/EE of NAFD distributed M-MIMO with low-resolution ADCs as
\begin{subequations}\label{equ:total}
\begin{align}
&f_1(\textbf{b})=\frac{T-\tau_1-\tau_2}{T}\left( \sum_{k=1}^{K_\mathrm{UL}}R_{\mathrm{UL},k}+\sum_{k=1}^{K_\mathrm{DL}}R_{\mathrm{DL},k}\right),\\
&f_2(\textbf{b})=\frac{W\cdot\left( \sum_{k=1}^{K_\mathrm{UL}}R_{\mathrm{UL},k}+\sum_{k=1}^{K_\mathrm{DL}}R_{\mathrm{DL},k}\right)}{P_\mathrm{TC}(\textbf{b})+P_\mathrm{T}+P_\mathrm{LP}+P_\mathrm{BH}},
\end{align}
\end{subequations}
where $\mathbf{b} = \{(b_1,\dots,b_{N_\mathrm{UL}});(b_1,\dots,b_{N_\mathrm{DL}});(b_1,\dots,b_{K_\mathrm{DL}})\}$ is the bit allocation scheme for UL RAUs, DL RAUs, and DL users. {The optimization problem is summarized as
\begin{equation}
\begin{aligned}
&\max\quad f=\left[ f_1(\mathbf{b}),f_2(\mathbf{b})\right] \\
& \begin{array}{r@{\quad}l@{}l@{\quad}l}
s.t.&{C1:}&B_\mathrm{U} \geq M\sum_{n=1}^{N+K_\mathrm{DL}} b_n, &b_n=1,2,... \\
&{C2:}&R_{\mathrm{UL},k}(\mathbf{b}) \geq R_{\mathrm{UL}}^\mathrm{min}, &k=1,...,K_\mathrm{UL} \\
&{C3:}&R_{\mathrm{DL},l}(\mathbf{b}) \geq R_{\mathrm{DL}}^\mathrm{min}, &k=1,...,K_\mathrm{DL} \\
&{C4:}& P_\mathrm{sum}(\mathbf{b}) \geq P_\mathrm{sum}(\mathbf{b}_\text{UQ,opt}), & \\
\end{array}
\end{aligned}
\end{equation}
{where $P_\mathrm{sum}(\mathbf{b})={P_\mathrm{TC}(\mathbf{b})+P_\mathrm{T}(\mathbf{b})+P_\mathrm{LP}(\mathbf{b})+P_\mathrm{BH}(\mathbf{b})}$. Constraint $C1$ expresses the limitation of the total bits of ADCs caused by the limited backhaul capacity, where $B_\mathrm{U}$ is the upper bound of the total ADC bits. Constraints $C2$ and $C3$ are the UL and DL average QoS requirements. Constraints $C4$ describes the EE lower bound as the power consumption corresponding to the optimal EE with guaranteed SE when quantized uniformly ($\mathbf{b}_\text{UQ,opt}$) for UL RAUs, DL RAUs, and DL users, respectively.}

\subsection{Problem Solution of ADC Bit Allocation}
Drawn lessons from the pioneer solutions which easily fall into local minima with slow convergence and high complexity when solving non-convex discrete optimization problems like Eq.~(\ref{equ:total}), DQN and NSGA-II is adopted for optimal solutions. 
\subsubsection{Problem Solution Based on DQN}
The solution process based on reinforcement learning (RL) methods includes agent, environment, reward, and action \cite{mnih2013playing}, which can be abstracted as a finite Markov decision process (MDP) \cite{sutton2018reinforcement}. Define the state and action at the $t$-th time-step as 
\begin{equation}
s_t \triangleq \mathbf{b}_t, a_t \triangleq \mathbf{b}^{\prime}_t,
\end{equation}
where $\mathbf{b}_t$ and $\mathbf{b}^{\prime}_t$ represent the quantization bit vector and the one-bit change at the $t$-th time-step. The reward is defined as
\begin{equation}
r_t\!=\!\frac{f_1(\mathbf{b})\!-\!\min\!f_1(\mathbf{b})}{\max\!f_1(\mathbf{b})\!-\!\min\!f_1(\mathbf{b})}\!+\!\frac{f_2(\mathbf{b})\!-\!\min\!f_2(\mathbf{b})}{\max\!f_2(\mathbf{b})\!-\!\min\!f_2(\mathbf{b})}\!-\!\widetilde{r},
\end{equation}
where $\widetilde{r}$ represents normalization constant. Concerning the Q-learning in RL, the optimization goal is maximizing the future returns $R_t$ by deciding the actions of each step in a certain way, i.e., $\max_{\pi} \mathbb{E}[R_{t}|s_t,a_t,\pi(a|s)]$. Hence, define the optimal action-value function as
\begin{equation}\label{equ:oav}
Q_i(s,a)=\mathbb{E}\left[r+\gamma\max_{a^{\prime}}Q_i(s^{\prime},a^{\prime})|(s,a)\right],
\end{equation}
where $\gamma$ is the discount factor and $ \pi(a|s)$ is the strategy to take action $a$ in the $s$ state. Meanwhile, Eq.~(\ref{equ:oav}) obeys the Bellman equation and guides value iteration with convergence boundary $i\leftarrow \infty, Q^{\ast}\leftarrow Q_i$. Unlike the typical linear methods, \cite{mnih2013playing} trains the neural network to extract advanced features with the following loss function 
\begin{equation}
\mathcal{L}_i(\theta_i)=\mathbb{E}_{s,a\sim \rho}\left[|y_i-Q(s,a;\theta_i)|^2\right],
\end{equation}
where $y_i=\mathbb{E}[r+\gamma\max_{a^{\prime}}Q(s^{\prime},a^{\prime};\theta_{i-1})|s,a]$ represents the target of the $i$-th iteration, and $\rho$ is the behaviour distribution, which adopt the $\epsilon$-greedy strategy (with a probability of $p_{gdy}=\epsilon$ taking a random action $a_\mathrm{rnd}$) to \bl{ensure} adequate exploration of state space. And then, calculate its gradient as
\begin{equation}\label{equ:SGD}
\nabla_{\theta_i}\mathcal{L}_i(\theta_i)=\mathbb{E}_{s,a\sim \rho}[(y_i-Q(s,a;\theta_i))\nabla_{\theta_i}Q(s,a;\theta_i)].
\end{equation}
Rather than the full expectation above, stochastic gradient descent (SGD) is routinely considered as a computationally-expedient alternative. To deal with the non-stationary distribution and high-correlation of data, Q-learning obtains the experience and pools them into replay memory $\mathcal{M}$, which is randomly sampled to finish the mini-batch SGD \cite{mnih2013playing}. The quantization bit allocation algorithm based on DQN is summarized in Algorithm~\ref{alg2} with $T_\mathcal{H}$ iterations.

\begin{algorithm}[htpb]\small
	\caption{Bit Allocation Algorithm Based on DQN} 
	\label{alg2}
	\begin{algorithmic}[1]
		\Require
		\State Initialize the state $s_0 $ and the replay memory $\mathcal{M}$
		\State Initialize action-value function $Q$ with random weights
		\Ensure 
		\State The state $s_{\max}$ with the optimal future reward
		\For{$t=1:T_\mathcal{H}$}
		\State Use $\epsilon$-greedy behavior policy to select action:
		\begin{equation}
		\nonumber
		a_t=\left\{
		\begin{aligned}
		&\max_{a} Q_t(s_t,a;\theta), &p_\mathrm{gdy}=1-\epsilon\\
		&a_\mathrm{rnd}, &p_\mathrm{gdy}=\epsilon
		\end{aligned}
		\right..
		\end{equation}
		\State Execute action $a_t$ and obtain the reward $r_t$
		\State Set the next stage $s_{t+1}=s_t\bigcup\{a_t\}$
		\State Store $\{s_t,a_t,r_t,s_{t+1}\}$ to $\mathcal{M}$
		\State Sample random minibatch of transitions from $\mathcal{M}$
		\State Calculate the target $y_i$ and perform SGD of Eq.~(\ref{equ:SGD})
		\EndFor
	\end{algorithmic}
\end{algorithm}

\subsubsection{Problem Solution Based on NSGA-II}
To analyze the problem more in-depth and choose quantization schemes more flexibly, we need to find the Pareto-optimal boundary of the bit optimization problem for ADC quantization, tricky for DQN. \bl{Multi-objective evolutionary algorithms have significant advantages when simultaneously solving all objectives in one simulation. Specifically, we employ the NSGA-II proposed in \cite{NSGA} with considerable reliability and effectiveness.} Detail steps based on DQN are summarized in Algorithm~\ref{alg2}.

\begin{algorithm}[htpb]\small
	\caption{Bit Allocation Algorithm Based on NSGA-II} 
	\label{alg1}
	\begin{algorithmic}[1]
		\Require 
		Randomly generate the initial value
		\Ensure
		Pareto-optimal solutions; SE and EE
		\Repeat 
		\State Use tournament selection to select a parent population
		\State Perform selection, mutation, and crossing
		\State Generate offspring population
		\State Calculate SE and EE of offspring population
		\State Merge the parent population and offspring populations
		\State Calculate SE and EE of merged population
		\State Sort the merged population based on non-dominating sorting and crowding distances
		\State Select a new initial population to replace the old one.
		\Until{The optimization conditions are satisfied or reach the maximum generation}
	\end{algorithmic}
\end{algorithm}
\begin{Rem} 
{When traditional methods deal with such nonconvex optimization problems, infinite local optimums exist in the set of feasible domains and convergence to the global optimum typically faces tremendous challenges since it has been proven as non-deterministic polynomial-time hard (NP-hard). The computational complexity of the DQN method is $\mathcal{O}(N_\mathbf{b}T^2_\mathcal{H})$; the complexity of NSGA-II method is $\mathcal{O}(N_\text{obj}N_\text{pop}^2)$, where $N_\mathbf{b}$ and $N_\text{obj}=2$, represent the number of low-resolution ADCs that need to be optimized and the objective functions, respectively. $N_\text{pop}$ is the population size. After random initialization, both methods can converge quickly to the corresponding Pareto-optimal points.}
\end{Rem}
%%%%%%%%%%%%%%%%%%%%%%%%%%%%%%%%%%%%%%%%%%%%%%%%%%%%%%%%%%%%%%%%%%%%%%%%%%%%%%%%%%%%%%%%%
%%%%%%%%%%%%%%%%%%%%%%%%%%%%%%%%%%%%%%%%%%%%%%%%%%%%%%%%%%%%%%%%%%%%%%%%%%%%%%%%%%%%%%%%%
%%%%%%%%%%%%%%%%%%%%%%%%%%%%%%%%%%%%%%%%%%%%%%%%%%%%%%%%%%%%%%%%%%%%%%%%%%%%%%%%%%%%%%%%%

\section{Numerical Results and Analysis}\label{sec:numerical result}

In this section, Monte-Carlo simulation is employed to verify the closed-form expression of SE in NAFD distributed M-MIMO with low-resolution ADCs, and the bit allocation algorithm is carried out for Pareto-optimal solutions with the tradeoff between the total SE and EE.
\subsection{Parameter Setting}

In the NAFD distributed M-MIMO system, the simulated area is limited to a circular area with a radius of $1$ km. All users and RAUs are randomly distributed, where $N_\mathrm{UL}=3, N_\mathrm{DL}=3, K_\mathrm{UL}=2, K_\mathrm{DL}=3$, and the minimum access distance to RAUs is $30$ m. \red{Assume that the number of antennas at each RAU is $M=10$, if not otherwise specified, and each RAU employs ADCs with the same resolution.} The path loss exponent is set to $\alpha_\mathrm{UL}=\alpha_\mathrm{DL}=3.7$, $\alpha_\mathrm{I}=3$. \bl{The DL and UL noise variances are $\sigma_\mathrm{DL}^2 = \sigma_\mathrm{UL}^2 = 1$ W.} The coherence time is $T=196$ symbols, the lengths of UL and DL pilot sequence are $\tau_1=K$ and $\tau_2$ = $K_\mathrm{DL}$, respectively. Hence, the number of symbols used for data transmission is ${T_\mathrm{data}}=T-\tau_1-\tau_2$, where \bl{the pre-log factor $\frac{T-\tau_1-\tau_2}{T}$ impacting the overall SE is considered.} Table~\ref{table2} presents all the parameters related to power consumption same as \cite{lv2018spectral}. {For the training details, the DQN is trained on the TensorFlow platform v1.14.0. The RMSProp algorithm is adopted with a mini-batch size of 32, and the learning rate is $0.01$, where $\gamma=0.9, \epsilon=0.9, \mathcal{M}=2000$, and the solution process needs $T_\mathcal{H}=1,000$ iterations for complete convergence. Both evaluate$\_$net and target$\_$net networks in DQN consist of fully connected DNN with 2 hidden layers with ReLU activation functions, where their dimensions are 9 and 18. The dimensions of the input and output layers of both evaluate$\_$net network and target$\_$net network are equal to 9 and 18. In NSGA-II, the population size is $N_\text{pop}=200$, and the max number of generations is $N_\text{gen}=300$. }

\begin{table}[htbp]
	\centering	
	\caption{Parameters of power consumption \cite{li2018spectral}.}	\label{table2}
	\begin{tabular}{cc}
		\hline
		Parameter & Value \\
		\hline
		UL transmission power $p_\mathrm{UL}$ & $500$ mW \\
		\hline
		DL transmission power $p_\mathrm{DL}$ & $500$ mW \\
		\hline
		UL pilot power $p_\mathrm{UP}$ & $500$ mW \\
		\hline
		DL pilot power $p_\mathrm{DP}$ &$ 1.00 $ W \\
		\hline
		Power consumption at RAU $P_\mathrm{RAU}$ & $100$ mW \\
		\hline
		Power consumption at user $P_\mathrm{UE}$ & $100$ mW \\
		\hline
		Local oscillator power consumption $P_\mathrm{SYN}$ & $1.00$ W \\
		\hline
		Calculation efficiency $L_\mathrm{RAU}$& $12.8$ Gflops/W\\
		\hline
		Amplifier efficiency $\xi$ & 0.4\\
		\hline
		Fixed power consumption $P_0$ & $0.825$ W\\
		\hline
		Dynamic power consumption $P_\mathrm{BT}$ & $0.25$ W/(Gbit/s)\\
		\hline
		ADC power coefficient $a_0,a_1$ & $10^{-4},0.02$\\
		\hline
	\end{tabular}
\end{table}

\subsection{Results and Discussion}
\subsubsection{Verification of Closed-Form Expressions}
{Considering the parameters mentioned above, assume all DL users are equipped with the same resolution ADC, and MRT pre-coding and ZFT pre-coding are considered. Based on Eq.~(\ref{Rdlpre}) in Theorem \ref{thm1} and Eq.~(\ref{Rdlpre2}) in Theorem \ref{thm2}, Fig.~\ref{DL_fitting} illustrates the relationship between the average DL SE and ADC bits $b$ of each DL user with MRT or ZFT beamforming training and with statistical CSI or estimated CSI.}
\begin{figure}[htbp] 
\centering 
\includegraphics[width=0.4\textwidth]{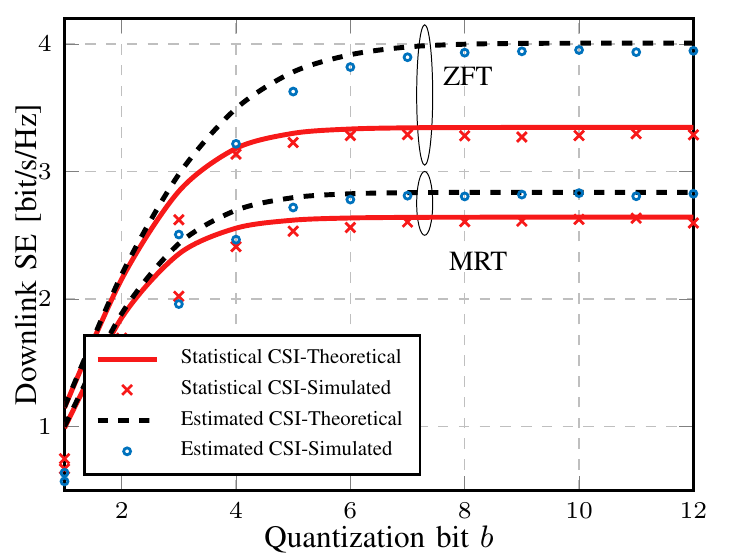}
\caption{Comparisons of DL SE with Different Beamforming Training Mechanisms and Different CSI Acquisition Methods.}
\label{DL_fitting}
\end{figure}
The curves in Fig.~\ref{DL_fitting} fit well, which validates the DL rate closed-form expression derived by Theorems \ref{thm1} and \ref{thm2} with the acceptable error caused by the Gamma approximation. \bl{By changing the values of different variables in the closed-form expressions, we can obtain good fitting results. Due to the space limitation, we provide the validation process only for the most important parameters in low-resolution studies, and others are similar. The subsequent studies are performed by closed-form expressions rather than time-consuming Monte Carlo simulations.}

No matter which pre-coding is used, the average SE increases as the quantization bits increases and then converges, which means that the upper bound of system performance can be achieved via a lower-resolution ADC. Meanwhile, the DL beamforming training for the estimated CSI brings more remarkable performance improvement than the statistical CSI. The DL SE of MRT and ZFT reaches the top when ADC resolution is $5$--$6$ bits and $7$--$8$ bits, respectively. Because ZFT pre-coding eliminates inter-user interference, it performs better than MRT regardless of whether the beamforming training estimation is performed. Moreover, when beamforming training is adopted, the CSI sensitivity of ZFT pre-coding leads to a more significant performance improvement, about $15.0\%$, compared to the MRT, only about $8.9\%$. In summary, ZFT is more sensitive to ADC resolution, which needs a higher resolution to achieve better performance. 

\begin{figure}[htbp] 
\centering 
\includegraphics[width=0.4\textwidth]{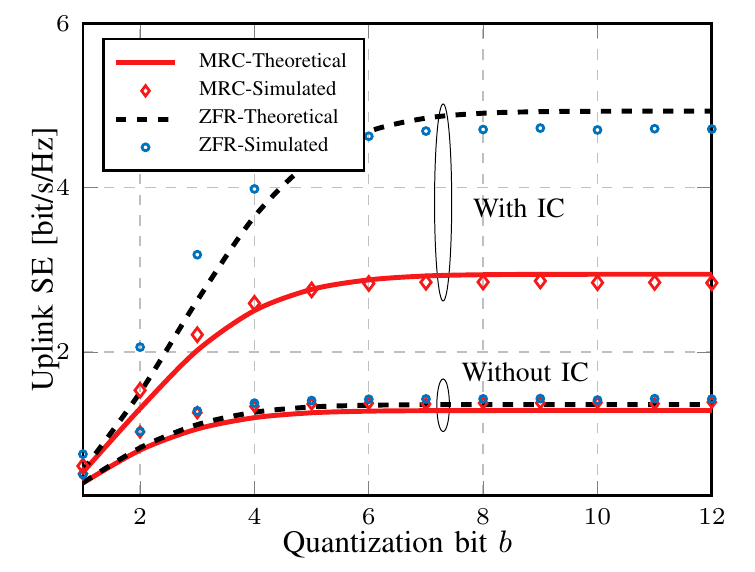}
\caption{Comparisons of UL SE with Different Receivers.} 
\label{UL_fitting}
\end{figure}

{Assume that the MRC and ZFR receivers are considered. Based on Eq.~(\ref{RulMRC}) in Theorem~\ref{thm3}, Fig.~\ref{UL_fitting} reveals the relationship between the average UL SE and the quantization bits $b$ of each RAU. It illustrates that the simulated result shows good agreement with the theoretical result, which verifies the accuracy of Theorem~\ref{thm3}.} The UL SE increases with the quantization bits and then converges to an upper bound, which shows that the upper bound can be reached by low-resolution ADCs, similar to the DL. Regardless of MRC or ZFR receiver, interference cancellation brings a tremendous performance improvement, proving the DL beamforming training is helpful for UL RAUs to eliminate interference between RAUs. The gap between MRC and ZFR receivers with CLI cancellation blows up since the strength of elimination determines the SE attenuation.

\subsubsection{Relationship Between the Total SE and EE}
Next, we explore the relationship between the total SE and EE in Eq.~(\ref{equ:total}). For the MRT pre-coding and MRC receiver, Fig.~\ref{sumSE-EE1} illustrates the total EE \textit{v.s.} SE with different numbers of quantization bits and antennas per RAU. Different lines represent different numbers of quantization bits $B\in\{b|4\leqslant b \leqslant 9,b\in\mathbb{N}\}$, and different points represent different numbers of antennas $M\in\{2m|3\leqslant m \leqslant 16,m\in\mathbb{N}\}$. For different lines, as the ADC accuracy increases, the SE gradually increases, and the EE increases and then decreases. For different points on the same line, as the number of antennas on each RAUs increases, the SE increases, and the EE increases and then decreases. It reveals that both EE and SE increase with the growth of the antenna and bit widths when the improvement of SE first dominates, followed by EE's prevailing.

\begin{figure}[htbp]
	\centering 
	\includegraphics[width=0.4\textwidth]{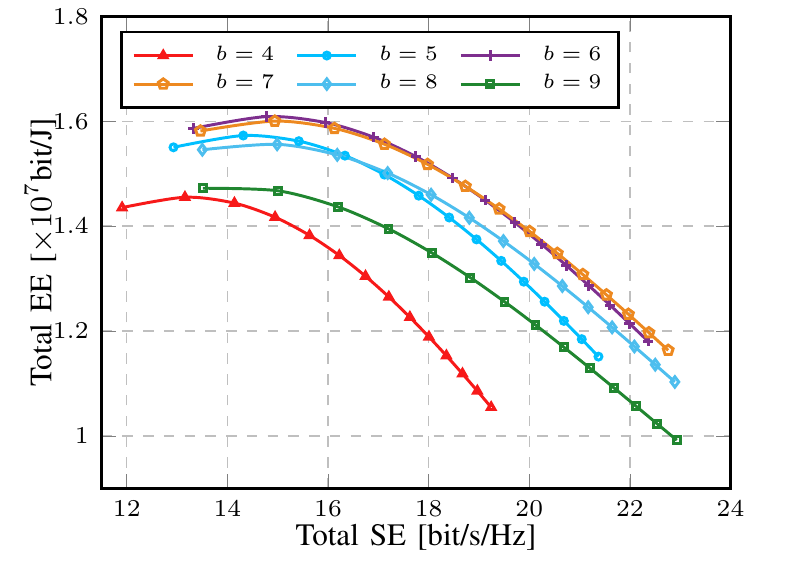} 		
	\caption{Tradeoff Between the Total SE and Total EE with Different Quantization Schemes and Antenna Numbers under MRT and MRC.} 
	\label{sumSE-EE1}
\end{figure}

\subsubsection{Bit Allocation Scheme}
As the ADC bitwidths increase, SE increases until saturation, but EE continues to decrease under mutual constraint and balance between SE and EE. To maximize the total SE and total EE, we design the bit allocation algorithm to obtain a tradeoff scheme. The QoS requirement \red{(further improvement of QoS reduces the number of solutions)} is $R_\mathrm{UL}^{\min}=R_\mathrm{DL}^{\min}= 1.5$ bit/s and the total ADC bit limit is $12MN+12K_\mathrm{DL}$. \bl{For MRT/MRC and ZFT/ZFR schemes, the power constraint is set to the consumption of the allocation schemes as $\{(7,7,7);(5,5,5);(6,6,6)\}$ and $\{(8,8,8);(1,1,1);(7,7,7)\}$, respectively. Fig.~\ref{nsga2} shows the tradeoff between total SE and EE, where the feasible bit allocation schemes for UL RAUs, DL RAU, and DL users (negligible with only one antenna), based on DQN or selected on the Pareto-optimal boundary by NSGA-II are $\{(6,7,7);(2,4,3);(6,6,6)\}$ or $\{(7,7,7);(5,4,5);(7,7,7)\}$ for MRC/MRT, and $\{(6,7,6);(1,1,1);(7,7,7)\}$ or $\{(6,7,6);(1,1,1);(7,4,7)\}$ for ZFT/ZFR. When quantization bitwidths become relatively large, a tiny increase results in a significant leap in the power consumption and a rapid decline of EE. Note that the SE-only or EE-only method is equivalent to linear weighting, and it provides only the endpoints of the Pareto front (EE of SE-only's optimum is too small without practical meaning, and the asymptotic lines are only provided), which leads to certain inflexibility for the tradeoff of MOOP. Further, all Pareto-optimal solutions with MRT/MRC and ZFT/ZFR schemes are obtained by NSGA-II, which validates the feasibility of the optimization methods and offers flexible ADC resolution choices for practical deployment.}

\begin{figure}[htbp]
	\centering 
	\subfigure[MRT Pre-coding and MRC Receiver]{
		\includegraphics[width=0.4\textwidth]{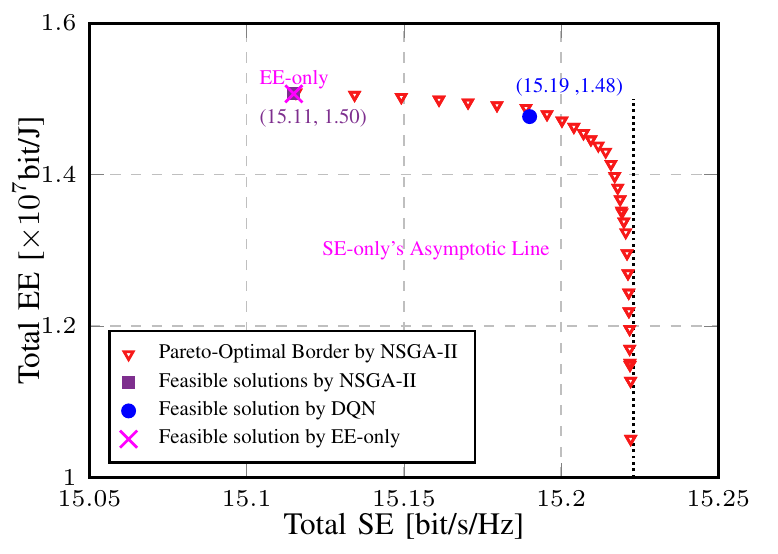}\label{fig:mrt}
		
	}
	\subfigure[ZFT Pre-coding and ZFR Receiver]{
		\includegraphics[width=0.4\textwidth]{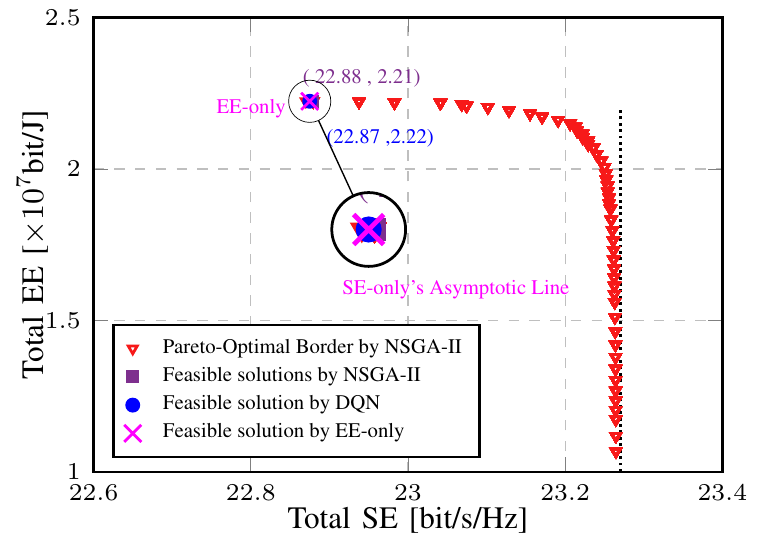}\label{fig:zf}	
	}
	\caption{\bl{SE and EE Tradeoff for Different Quantization Schemes.}} 
	\label{nsga2}
\end{figure}

\vspace{-0.5cm}
\section{Conclusions}\label{sec:conclusion}
This paper analyzed the NAFD distributed M-MIMO systems with low-resolution ADCs. Two-stage channel estimation, including pilot estimation and beamforming training, was adopted to complete the CLI elimination and obtain the effective CSI for signal detection and interference cancellation. With the estimated CSI, the closed-form expressions for the achievable rates of the UL and DL, as well as EE, were derived. Further, a MOOP maximizing SE and EE under per-user QoS and total bit constraints was formulated to design an efficient ADC bit allocation scheme. The simulation results verified the theoretical analysis. Low-resolution ADCs significantly improved EE in NAFD distributed M-MIMO. Interference cancellation via beamforming training enhanced the UL SE with different receivers under low-resolution ADCs, and channel estimation of DL users refined the DL SE, especially for the ZFT pre-coding and ZFR receiver. DQN and NSGA-II were selected to solve the bit allocation problem, which obtained Pareto-optimal solutions for ADC resolution providing great flexibility for system development.

\vspace{-0.5cm}
\begin{appendices}
\section{Mathematical Results}\label{appendix 1}
\begin{Lemma}\label{Lemma 1}
	$\bm{x}\in\mathbb{C}^{m\times 1}$ distributes as $\bm{x}_i\sim \mathcal{CN}(0,\sigma_i^2\textbf{I})$ \cite{li2017downlink},
	\begin{subequations}
	\begin{align}
	\bm{x} _i^\mathrm{H}\bm{x}_i &\sim \Gamma(m, \sigma_i^2),\\
	\sum_i\bm{x}_i^\mathrm{H}\bm{x}_i &\sim \Gamma\left(m\frac{\left(\sum_i\sigma_i^2\right)^2}{\sum_i\sigma_i^4}, \frac{\sum_i\sigma_i^4}{\sum_i\sigma_i^2}\right).
	\end{align}
	\end{subequations}
\end{Lemma}

\begin{Lemma}
	$\bm{x}\in\mathbb{C}^{m\times 1}$ satisfies $\bm{x}^\mathrm{H}\bm{x}\sim\Gamma(k,\theta)$, and its elements are independent non-identical distributions, when $\bm{x}$ is projected to an $s$-dimensional subspace, the distribution of projection power can be approximated as $\Gamma(sk/m,\theta)$ \cite{li2017downlink}.
\end{Lemma}

\section{Proof of Proposition 1}\label{appendix 2}
\begin{proof}
Consider MRT pre-coding as $\mathbf{w}_i=\frac{\hat{\mathbf{g}}_{\mathrm{DL},i}}{\left\| \hat{\mathbf{g}}_{\mathrm{DL},i}\right\| }$, $
\mathbb{E}\left[ \mu_{k,i}\right]=\mathbb{E}\left[ \mathbf{g}_{\mathrm{DL},k}\mathbf{w}_i\right]$. When $i=k$, $\mathbb{E}\left[ \mu_{k,i}\right]=\mathbb{E}\left[ \left\| \hat{\mathbf{g}}_{\mathrm{DL},k}\right\|\right] $. According to the Lemmas introduced in Appendix~\ref{appendix 1}, we have
\begin{subequations}
\begin{align}
&\left\| \hat{\mathbf{g}}_{\mathrm{DL},n,k}\right\|^2\sim\Gamma\left( M,\beta_{\mathrm{DL},n,k}\right)\\
&\left\| \hat{\mathbf{g}}_{\mathrm{DL},k}\right\|^2=\sum_{n=1}^{N_{\mathrm{DL}}} \left\| \hat{\mathbf{g}}_{\mathrm{DL},n,k}\right\|^2\sim\Gamma\left( k_{\mathrm{DL},k},\theta_{\mathrm{DL},k}\right)
\end{align}
\end{subequations}
where
\begin{small}
\begin{subequations}
\begin{align}
&k_{\mathrm{DL},k}=\frac{M\left( \sum_{n=1}^{N_{\mathrm{DL}}}\beta_{\mathrm{DL},n,k}\right) ^2}{\sum_{n=1}^{N_{\mathrm{DL}}}\beta_{\mathrm{DL},n,k}^2},\theta_{\mathrm{DL},k}=\frac{\sum_{n=1}^{N_{\mathrm{DL}}}\beta_{\mathrm{DL},n,k}^2}{\sum_{n=1}^{N_{\mathrm{DL}}}\beta_{\mathrm{DL},n,k}},\\
&\tilde{k}_{\mathrm{DL},k}=\frac{M\left( \sum_{n=1}^{N_{\mathrm{DL}}}\eta_{\mathrm{DL},n,k}\right) ^2}{\sum_{n=1}^{N_{\mathrm{DL}}}\eta_{\mathrm{DL},n,k}^2},\tilde{\theta}_{\mathrm{DL},k}=\frac{\sum_{n=1}^{N_{\mathrm{DL}}}\eta_{\mathrm{DL},n,k}^2}{\sum_{n=1}^{N_{\mathrm{DL}}}\eta_{\mathrm{DL},n,k}}.
\end{align}
\end{subequations}
\end{small}

From the relationship between Gamma and Nakagami distribution, $\left\| \hat{\mathbf{g}}_{\mathrm{DL},k}\right\| \sim \mathrm{Nakagami}(k_{\mathrm{DL},k},k_{\mathrm{DL},k}\theta_{\mathrm{DL},k})$, 
\begin{equation}
\mathbb{E}\left[ \left\| \hat{\mathbf{g}}_{\mathrm{DL},k}\right\|\right]=\frac{\Gamma\left( k_{\mathrm{DL},k}+{1}/{2}\right) }{\Gamma\left( k_{\mathrm{DL},k}\right)}\theta_{\mathrm{DL},k}^{1/2}.
\end{equation}

When $i\neq k$, $\hat{\mathbf{g}}_{\mathrm{DL},k}$ is independent of $\hat{\mathbf{g}}_{\mathrm{DL},i}$, $\mathbb{E}\left[ \mu_{k,i}\right]=0$. Combining the above two cases, 
\begin{equation}\label{Emu}
\mathbb{E}\left[ \mu_{k,i}\right]=\left\{
\begin{aligned}
&\frac{\Gamma\left( k_{\mathrm{DL},k}+{1}/{2}\right) }{\Gamma\left( k_{\mathrm{DL},k}\right)}\theta_{\mathrm{DL},k}^{1/2}, &i=k\\
&0,&i\neq k
\end{aligned}
\right..
\end{equation}

Moreover, we have
\begin{equation}\label{covmur}
\begin{aligned}
\mathrm{cov}(\mu_{k,i},\tilde{\bm{r}}_{e,k}^{p,i})
&=\mathbb{E}\left[ \mu_{k,i}\tilde{\bm{r}}_{e,k}^{p,i}\right]-\mathbb{E}\left[ \mu_{k,i}\right]\mathbb{E}\left[ \tilde{\bm{r}}_{e,k}^{p,i}\right]\\
&=\xi_k\sqrt{p_{\mathrm{DP}}}\left( \mathbb{E}\left[ \mu_{k,i}^2\right]-\mathbb{E}\left[ \mu_{k,i}\right]^2\right),
\end{aligned}			 
\end{equation}
where
\begin{equation}\label{mu^2}
\begin{aligned}
&\mathbb{E}\left[ \mu_{k,i}^2\right]=\mathbb{E}\left[ \left\| \hat{\mathbf{g}}_{\mathrm{DL},k}\mathbf{w}_i\right\|^2\right] +\mathbb{E}\left[ \left\| \tilde{\mathbf{g}}_{\mathrm{DL},k}\mathbf{w}_i\right\|^2\right] .
\end{aligned}
\end{equation}
	
If RAUs use MRT pre-coding, the pre-coding vector is placed in the $N_{\mathrm{DL}}M$-dimensional space. If RAUs use ZFT pre-coding, the pre-coding vector is placed in $N_{\mathrm{DL}}M-K_{\mathrm{DL}}+1$ dimensional space. With the Lemmas in Appendix \ref{appendix 1},
\begin{equation}
\left\{
\begin{aligned}
&\left\| \hat{\mathbf{g}}_{\mathrm{DL},k}\mathbf{w}_i\right\|^2\sim \Gamma\left( k_{\mathrm{DL},k},\theta_{\mathrm{DL},k}\right), &i=k\\
&\left\| \hat{\mathbf{g}}_{\mathrm{DL},k}\mathbf{w}_i\right\|^2\sim \Gamma\left( \frac{1}{N_{\mathrm{DL}}M}k_{\mathrm{DL},k},\theta_{\mathrm{DL},k}\right), &i\neq k\\
&\left\| \tilde{\mathbf{g}}_{\mathrm{DL},k}\mathbf{w}_i\right\|^2\sim \Gamma\left( \frac{1}{N_{\mathrm{DL}}M}\tilde{k}_{\mathrm{DL},k},\tilde{\theta}_{\mathrm{DL},k}\right), &i\neq k\\
\end{aligned}
\right..
\end{equation}
Therefore,
\begin{equation}\label{Emu^2}
\mathbb{E}\left[ \mu_{k,i}^2\right]=\left\{
\begin{aligned}
&k_{\mathrm{DL},k}\theta_{\mathrm{DL},k}+\frac{\tilde{k}_{\mathrm{DL},k}\tilde{\theta}_{\mathrm{DL},k}}{N_{\mathrm{DL}}M}, &i=k\\
&\frac{ k_{\mathrm{DL},k}\theta_{\mathrm{DL},k}+\tilde{k}_{\mathrm{DL},k}\tilde{\theta}_{\mathrm{DL},k}}{N_{\mathrm{DL}}M}, &i\neq k
\end{aligned}
\right..
\end{equation}
Substituting Eqs.~(\ref{Emu^2}) and (\ref{Emu}) into Eq.~(\ref{covmur}), we have
\begin{equation}\label{covmur2}
\mathrm{cov}(\mu_{k,i},\tilde{\bm{r}}_{e,k}^{p,i})=\left\{
\begin{aligned}
&\xi_k\sqrt{p_{\mathrm{DP}}}\bar{\chi}_k^\mathrm{MR}, &i=k\\
&\xi_k\sqrt{p_{\mathrm{DP}}}{\chi}_k^\mathrm{MR}, &i\neq k
\end{aligned}
\right..
\end{equation}
Meanwhile,
\begin{equation}
\begin{aligned}
&\mathrm{cov}(\tilde{\bm{r}}_{e,k}^{p,i},\tilde{\bm{r}}_{e,k}^{p,i})\\
&=\mathbb{E}[\xi_k^2p_{\mathrm{DP}}\mu_{k,i}^2\!+\!\xi_k^2\mathrm{n}_{\mathrm{DP},i}^2\!+\!\tilde{\mathrm{n}}_{\mathrm{DP},i}^2]\!-\!\mathbb{E}\left[ \xi_k\sqrt{p_{\mathrm{DP}}}\mu_{k,i}\right]^2\\
&=\xi_k\left( p_{\mathrm{DP}} \bar{\chi}_k^\mathrm{MR}+\sigma_{\mathrm{DP}}^2\right)\!\!+\!\underbrace{(\xi_k\!-\!\xi_k^2)p_{\mathrm{DP}}(E_k^\mathrm{MR})^2}_{\text{when} \,\,i=k},
\end{aligned}
\end{equation}	
where
\begin{small}
\begin{equation}
\begin{aligned}
&\mathbb{E}\left[ \tilde{\mathrm{n}}_{\mathrm{DP},i}^2\right]=\xi_k(1-\xi_k)\mathbb{E}\left[ (\sqrt{p_{\mathrm{DP}}}\mu_{k,i}+\mathrm{n}_{\mathrm{DP},i})^2\right]\\
&=\left\{
\begin{aligned}
&\xi_k(1-\xi_k)\left( p_{\mathrm{DP}}\left[ \bar{\chi}_k^\mathrm{MR}+(E_k^\mathrm{MR})^2\right] +\sigma_{\mathrm{\mathrm{DP}}}^2\right), &i=k\\
&\xi_k(1-\xi_k)\left( p_{\mathrm{DP}}{\chi}_k^\mathrm{MR} +\sigma_{\mathrm{DP}}^2\right), &i\neq k	
\end{aligned} 
\right..
\end{aligned}
\end{equation}
\end{small}

Finally, we obtain
\begin{small}
\begin{equation}\label{Er}
\begin{split}
\mathbb{E}\left[\tilde{\bm{r}}_{e,k}^{p,i}\right]&=\xi_k\sqrt{p_{\mathrm{DP}}}\mathbb{E}\left[ \mu_{k,i}\right] =\left\{
\begin{aligned}
&\xi_k\sqrt{p_{\mathrm{DP}}}E_k^\mathrm{MR}, &i=k\\
&0, &i\neq k
\end{aligned}
\right..
\end{split}
\end{equation}
\end{small}

Next is a similar proof of the ZFT. There is only a small difference between ZFT and MRT. If RAUs use ZFT pre-coding, the pre-coding vector is placed in $N_{\mathrm{DL}}M-K_{\mathrm{DL}}+1$ dimensional space. When $i=k$, $\left\| \hat{\mathbf{g}}_{\mathrm{DL},k}\mathbf{w}_i\right\|^2\sim\Gamma\left( t_{\mathrm{DL}}k_{\mathrm{DL},k},\theta_{\mathrm{DL},k}\right) $, and
\begin{equation}
\mathbb{E}\left[ \mu_{k,i}\right]=\left\{
\begin{aligned}
&\frac{\Gamma\left( t_{\mathrm{DL}}k_{\mathrm{DL},k}+{1}/{2}\right) }{\Gamma\left( t_{\mathrm{DL}}k_{\mathrm{DL},k}\right)}\theta_{\mathrm{DL},k}^{1/2}, &i=k\\
&0,&i\neq k
\end{aligned}
\right..
\end{equation}
Other than that, it is the same as MRT. Therefore, we omit the proof of ZFT due to the limitation of space.

\end{proof}
	
\vspace{-0.5cm}
\section{Proof of Proposition 2}\label{appendix 3}
\begin{proof}
Since $\mathbf{G}_\mathrm{I}$ and $\mathbf{w}_{i}$ are independent, $\mathrm{cov}(\mathbf{f}_i,\tilde{\mathbf{y}}_i^p)\!=\!\mathbb{E}[ \mathbf{f}_i(\tilde{\mathbf{y}}_i^p)^\mathrm{H}]-\mathbb{E}\left[ \mathbf{f}_i\right]\mathbb{E}\left[\tilde{\mathbf{y}}_i^p\right]{=}\sqrt{p_{\mathrm{DP}}}\mathbf{A}\mathbb{E}[ \mathbf{f}_i\mathbf{f}_i^\mathrm{H}]$, where $\mathbb{E}\left[ \mathbf{f}_i\right]=\mathbf{0}$, $\mathbb{E}\left[ \mathbf{y}_i^p\right]=\mathbf{0}$. $\mathbf{g}_{\mathrm{I},k}$ is the $k$-th line of the interference channel $\mathbf{G}_\mathrm{I}$. Next, $\mathbb{E}\left[ \mathbf{g}_{\mathrm{I},k}\mathbf{w}_i\mathbf{w}_i^\mathrm{H}\mathbf{g}_{\mathrm{I},j}^\mathrm{H}\right]=0$, when $j\neq k$. Otherwise, $
\mathbb{E}\left[ \mathbf{g}_{\mathrm{I},k} \mathbf{w}_i\mathbf{w}_i^\mathrm{H}\mathbf{g}_{\mathrm{I},j}^\mathrm{H}\right]=\frac{1}{N_{\mathrm{DL}}}\sum_{n=1}^{N_{\mathrm{DL}}}\lambda_{\mathrm{I},k,n}.$
And then, 
\begin{footnotesize}
\begin{equation}\label{covyy}
\begin{aligned}
&\mathrm{cov}(\tilde{\mathbf{y}}_i^p,	\tilde{\mathbf{y}}_i^p)=\mathbb{E}\left[ \tilde{\mathbf{y}}_i^p(\tilde{\mathbf{y}}_i^p)^\mathrm{H}\right] \\
&=p_{\mathrm{DP}}\mathbf{A}^2(\mathbb{E}\left[ \mathbf{f}_i\mathbf{f}_i^\mathrm{H}\right]\!+\!\mathbf{I})\!+\!\mathbf{A}(\mathbf{I}\!-\!\mathbf{A})\mathrm{diag}\left( p_{\mathrm{DP}}\mathbb{E}\left[ \mathbf{f}_i\mathbf{f}_i^\mathrm{H}\right]\!+\!\sigma_{\mathrm{UP}}^2\mathbf{I}\right) \\
&=p_{\mathrm{DP}}\mathbf{A}\bm{\Lambda}+\mathbf{A}\sigma_{\mathrm{UP}}^2.
\end{aligned}
\end{equation}
\end{footnotesize}
 
Because $\mathrm{cov}(\tilde{\mathbf{y}}_i^p, \tilde{\mathbf{y}}_i^p)^{-1/2}\tilde{\mathbf{y} }_i^p$ can be equivalent to a small-scale fading following the distribution of i.i.d. $\mathcal{CN}(0,\mathbf{I})$, then $\mathrm{cov}(\mathbf{f}_i,\tilde {\mathbf{y}}_i^p)\mathrm{cov}(\tilde{\mathbf{y}}_i^p, \tilde{\mathbf{y}}_i^p)^{-1/2} $ can be equivalent to large-scale fading. After simple derivation, Eq.~(\ref{delta}) is derived.
\end{proof}

\section{Proof of Theorem 1}\label{appendix 4}
\begin{proof}
The derivation of DL rate expression is similar for MRT and ZFT pre-coding. Firstly, for Eq.~(\ref{Rdl}), we use common approximation as $\mathbb{E}\left[ \mathrm{log}_2\left( 1+\frac{{X}}{{Y}}\right) \right] \approx\mathrm{log}_2\left( 1+\frac{\mathbb{E}\left[ {X}\right] }{\mathbb{E}\left[ {Y}\right] }\right)$, which is tight when the number of antennas tends to infinity, and $X,Y$ are sums of non-negative random variables and converge to their mean due to the law of large numbers. Thus,
\begin{small}
\begin{equation}\label{Rdl2}
R_{\mathrm{DL},k}\!\approx\!\mathrm{log}_2\left( 1\!+\!\frac{\mathbb{E}\left[ \left| \xi_k\sqrt{p_{\mathrm{DL}}}\hat{\mu}_{k,k}s_k\right|^2\right] }{\mathbb{E}\left[ \phi_{\mathrm{DL},k}\!+\!\psi_{\mathrm{DL},k}\!+\!\rho_{\mathrm{DL},k}\!+\!\widetilde{\mathbf{C}}_{n_{\mathrm{DL},k}}\right] }\right). 
\end{equation}
\end{small}
% \end{footnotesize}

Calculate its numerator as
% \begin{small}
\begin{equation}\label{dlnume}
\begin{aligned}
&\mathbb{E}\left[ \left| \xi_k\sqrt{p_{\mathrm{DL}}}\hat{\mu}_{k,k}s_k\right|^2\right]=\xi_k^2p_{\mathrm{DL}}\mathbb{E}\left[ \left| \hat{\mu}_{k,k}\right| ^2\right] \\
% &=\xi_k^2p_{\mathrm{DL}}\mathbb{E}\left[ \mathbb{E}\left[ \mu_{k,k}\right] +\frac{\mathrm{cov}(\mu_{k,k},	\tilde{\bm{r}}_{e,k,k}^{p})}{\mathrm{cov}(\tilde{\bm{r}}_{e,k,k}^{p},	\tilde{\bm{r}}_{e,k,k}^{p})}\left(\tilde{\bm{r}}_{e,k}^{p,i}-\mathbb{E}\left[ 	\tilde{\bm{r}}_{e,k,k}^{p}\right] \right) \right] ^2\\
% &=\xi_k^2p_{\mathrm{DL}}\left[ \left( \mathbb{E}\left[ \mu_{k,k}\right]\right)^2+\frac{ \mathrm{cov}^2(\mu_{k,k},	\tilde{\bm{r}}_{e,k,k}^{p})}{\mathrm{cov}(\tilde{\bm{r}}_{e,k,k}^{p},	\tilde{\bm{r}}_{e,k,k}^{p})} \right] \\
&=\xi_k^2p_{\mathrm{DL}}\left[ \left( E_k^\mathrm{pre}\right)^2+\frac{p_{\mathrm{DP}}(\bar{\chi}_k^\mathrm{pre})^2}{p_{\mathrm{DP}}\widetilde{\chi}_k^\mathrm{pre}+\sigma_{\mathrm{DP}}^2} \right]. 
\end{aligned}
\end{equation}
% \end{small}

Then, calculate the first term of denominator
\begin{equation}\label{dldeno1}
\begin{aligned}
&\mathbb{E}\left[ \phi_{\mathrm{DL},k}\right]=\xi_k^2\sum\nolimits_{i\neq k}p_{\mathrm{DL}}\mathbb{E}\left[ \left| \hat{\mu}_{k,i}\right| ^2\right]\\
% &=\xi_k^2\sum_{i\neq k}p_{\mathrm{DL}}\left[ \left( \mathbb{E}\left[ \mu_{k,i}\right]\right)^2+\frac{ \mathrm{cov}^2(\mu_{k,i},\tilde{\bm{r}}_{e,k}^{p,i})}{\mathrm{cov}(\tilde{\bm{r}}_{e,k}^{p,i},\tilde{\bm{r}}_{e,k}^{p,i})} \right]\\
&=\xi_k^2\sum\nolimits_{i\neq k}p_{\mathrm{DL}}\frac{\xi_kp_{\mathrm{DP}}\left( \chi_k^\mathrm{pre}\right)^2 }{p_{\mathrm{DP}}\chi_k^\mathrm{pre}+\sigma_{\mathrm{DP}}^2}.
\end{aligned}
\end{equation}

The second term of denominator is obtaine as 
\begin{equation}\label{dldeno2}
\begin{aligned}
&\mathbb{E}\left[ \psi_{\mathrm{DL},k}\right] =\mathbb{E}\left[ \left| \xi_k\sum\nolimits_{i=1}^{K_{\mathrm{DL}}}\sqrt{p_{\mathrm{DL}}}\tilde{\mu}_{k,i}s_i\right|^2\right] \\
&=\xi_k^2\sum\nolimits_{i=1}^{K_{\mathrm{DL}}}{p_{\mathrm{DL}}}\left( \mathbb{E}\left[ \left| {\mu}_{k,i}\right| ^2\right]-\mathbb{E}\left[ \left| \hat{\mu}_{k,i}\right| ^2\right]\right). 
\end{aligned}
\end{equation}

When $i=k$, denote $\Delta \mathbb{E}[\mu]= \mathbb{E}\left[ \left| {\mu}_{k,k}\right| ^2\right]-\mathbb{E}\left[ \left| \hat{\mu}_{k,k}\right| ^2\right]$
\begin{equation}\label{dldeno21}
\Delta \mathbb{E}[\mu]=\frac{\left( 1-\xi_k\right)p_{\mathrm{DP}}\bar{\chi}_k^\mathrm{pre}\left( E_k^\mathrm{pre}\right)^2+\bar{\chi}_k^\mathrm{pre}\sigma_{\mathrm{DP}}^2 }{p_{\mathrm{DP}}\bar{\chi}_k^\mathrm{pre} \! \!+(1 \! -\!\xi_k)p_{\mathrm{DP}}(E_k^\mathrm{pre})^2\!+\sigma_{\mathrm{DP}}^2}.
\end{equation}

When $i\neq k$,
\begin{equation}\label{dldeno22}
\Delta \mathbb{E}[\mu]=\frac{\left( 1-\xi_k\right)p_{\mathrm{DP}}\left( \chi_k^\mathrm{pre}\right) ^2+\sigma_{\mathrm{\mathrm{DP}}}^2\chi_k^\mathrm{pre} }{p_{\mathrm{DP}}\chi_k^\mathrm{pre}+\sigma_{\mathrm{DP}}^2}.
\end{equation}

Substituting Eqs.~(\ref{dldeno21}) and (\ref{dldeno22}) into Eq.~(\ref{dldeno2}), we obtain
\begin{equation}\label{dldeno2final}
\begin{aligned}
&\mathbb{E}\left[ \psi_{\mathrm{DL},k}\right]=\xi_k^2{p_{\mathrm{DL}}}\frac{\left( 1-\xi_k\right)p_{\mathrm{DP}}\bar{\chi}_k^\mathrm{pre}\left( E_k^\mathrm{pre}\right)^2+\bar{\chi}_k^\mathrm{pre}\sigma_{\mathrm{DP}}^2 }{p_{\mathrm{DP}}\bar{\chi}_k^\mathrm{pre}+(1-\xi_k)p_{\mathrm{DP}}(E_k^\mathrm{pre})^2+\sigma_{\mathrm{DP}}^2}\\
&+\xi_k^2\sum_{i\neq k}p_{\mathrm{DL}}\frac{\left( 1-\xi_k\right)p_{\mathrm{DP}}\left( \chi_k^\mathrm{pre}\right) ^2+\sigma_{\mathrm{DP}}^2\chi_k^\mathrm{pre} }{p_{\mathrm{DP}}\chi_k^\mathrm{pre}+\sigma_{\mathrm{DP}}^2}.
\end{aligned}
\end{equation}
We get the third term of denominator as
\begin{equation}\label{dldeno3}
\begin{aligned}
\mathbb{E}[| \xi_k\sum_{j=1}^{K_{\mathrm{UL}}}\sqrt{p_{\mathrm{UL}}}u_{\mathrm{I},k,j}x_j|^2]\!=\!\xi_k^2 \sum_{j=1}^{K_{\mathrm{UL}}}p_{\mathrm{UL}}\lambda_{\mathrm{I},k,j},
\end{aligned}
\end{equation}
and the last term of denominator as
\begin{equation}\label{dldeno4}
\mathbf{\widetilde{C}}_{n_{\mathrm{DL},k}}=\xi_k\left( 1-\xi_k\right)\mathbb{E}\left[ \mathrm{diag}\left( r_{\mathrm{DL},k}r_{\mathrm{DL},k}^\mathrm{H}\right) \right],
\end{equation}
where $\mathbb{E}\left[ \mathrm{diag}\left( r_{\mathrm{DL},k}r_{\mathrm{DL},k}^\mathrm{H}\right) \right]=\sum_{i=1}^{K_{\mathrm{DL}}}{p_{\mathrm{DL}}}\mathbb{E}\left[ \left| {\mu}_{k,i}\right| ^2\right]+\sum_{j=1}^{K_{\mathrm{UL}}}{p_{\mathrm{UL}}}\mathbb{E}\left[ \left| u_{\mathrm{I},k,j}\right| ^2\right]+\sigma_{\mathrm{DL}}^2 $. Substituting Eqs.~(\ref{dlnume}-\ref{dldeno4}) into Eq.~(\ref{Rdl}), we can obtain the result in Eqs.~(\ref{Rdlpre}) and (\ref{dlparts}).

\end{proof}
\section{Proof of Theorem 3}\label{appendix 5}
\begin{proof}
Firstly, dividing the numerator and denominator of the SINR in Eq.~(\ref{Rul}) by $\left\| \mathbf{v}_k\right\|^2 $, the rate expression is approximated as Eq.~(\ref{Rulkapprox}) at the top of next page. For MRC receiver, substitute $\mathbf{v}_k=\hat{\mathbf{g}}_{\mathrm{UL},k}$ into Eq.~(\ref{Rulkapprox}). We have the numerator as $M\Xi_{n,k} =M\sum_{n=1}^{N_{\mathrm{UL}}}\alpha_n^2\beta_{\mathrm{UL},n,k}$, where 
\begin{small}
\begin{subequations}
\begin{align}
&\left\| \mathbf{A}\hat{\mathbf{g}}_{\mathrm{UL},k}\right\|^2=\sum_{n=1}^{N_{\mathrm{UL}}}\left\| \alpha_n\hat{\mathbf{g}}_{\mathrm{UL},n,k} \right\| ^2\sim\Gamma\left(M \hat{k}_{\mathrm{UL},k},\hat{\theta}_{\mathrm{UL},k}\right), \\
&\hat{k}_{\mathrm{UL},k}=\frac{\Xi_{n,k}^2 }{\sum_{n=1}^{N_{\mathrm{UL}}}\alpha_n^4\beta_{\mathrm{UL},n,k}^2},	\hat{\theta}_{\mathrm{UL},k}=\frac{\sum_{n=1}^{N_{\mathrm{UL}}}\alpha_n^4\beta_{\mathrm{UL},n,k}^2}{\Xi_{n,k}}.
\end{align}
\end{subequations}
\end{small}
\begin{figure*}[htpb]
\begin{small}
		\begin{equation}\label{Rulkapprox}
		\begin{aligned}
		&R_{\mathrm{UL},k}\approx \mathrm{log}_2\left( 1+\frac{\mathbb{E}\left[ p_{\mathrm{UL}}\left| \mathbf{v}_k^\mathrm{H}/\left\| \mathbf{v}_k\right\|\mathbf{A}\hat{\mathbf{g}}_{\mathrm{UL},k}\right|^2 \right] } {\mathbb{E}\left[ \phi_{\mathrm{UL},k}^{'}+\psi_{\mathrm{UL},k}^{'}+\rho_{\mathrm{UL},k}^{'}+\left\|\mathbf{v}_k^\mathrm{H}/\left\| \mathbf{v}_k\right\|\mathbf{A} \right\|^2\sigma_{\mathrm{UL}}^2+\mathbf{v}_k^\mathrm{H}/\left\| \mathbf{v}_k\right\|\mathbf{R}_{\mathbf{n}_{q,\mathrm{UL}}}\mathbf{v}_k^\mathrm{H}/\left\| \mathbf{v}_k\right\|\right] }\right).
		\end{aligned}
		\end{equation}
\end{small}
	\hrulefill
\end{figure*}

Next, for the first term of denominator, the receiver vectors with MRC are placed in $N_{\mathrm{UL}}M$-dimensional space. Any vector independent of the channel vector is placed in a one-dimensional space. $\hat{\mathbf{g}}_{\mathrm{UL},i}$ is independent of $\hat{\mathbf{g}}_{\mathrm{UL},k}$. Hence,
\begin{equation}\label{Ruldeno1}
\mathbb{E}\left[ \phi_{\mathrm{UL},k}^\mathrm{MRC'}\right]=\frac{1}{N_{\mathrm{UL}}}\sum_{n=1}^{N_{\mathrm{UL}}}p_{\mathrm{UL}}\alpha_n^2\beta_{\mathrm{UL},n,i}.
\end{equation}

And the other terms of denominator can be expressed as
\begin{small}
\begin{equation}\label{Ruldeno4}
\begin{aligned}
&\mathbb{E}\left[ \psi_{\mathrm{UL},k}^\mathrm{MRC'}\right] 
=\frac{1}{N_{\mathrm{UL}}}\sum_{i=1}^{K_{\mathrm{UL}}}p_{\mathrm{UL}}\sum_{n=1}^{N_{\mathrm{UL}}}\alpha_n^2\eta_{\mathrm{UL},n,i}.\\
&\mathbb{E}\left[ \rho_{\mathrm{UL},k}^\mathrm{MRC'}\right] 
=\frac{1}{N_{\mathrm{UL}}}\sum_{i=1}^{K_{\mathrm{UL}}}p_{\mathrm{UL}}\sum_{n=1}^{N_{\mathrm{UL}}}\alpha_n^2\rho_{\mathrm{UL},n,j}^2.\\
&\mathbb{E}\left[ \left\|\frac{\mathbf{v}_k^\mathrm{H}}{\left\| \mathbf{v}_k\right\|}\mathbf{A} \right\|^2\sigma_{\mathrm{UL}}^2\right]=\frac{1}{N_{\mathrm{UL}}}\sigma_{\mathrm{UL}}^2\sum_{n=1}^{N_{\mathrm{UL}}}\alpha_n^2.
\end{aligned}
\end{equation}
\end{small}

Finally, substitute Eqs.~(\ref{Ruldeno1}-\ref{Ruldeno4}) into Eq.~(\ref{Rulkapprox}) for the UL rate closed-form expression under MRC receiver. With ZFR, the receiver vector is placed in $N_{\mathrm{UL}}M-K_{\mathrm{UL}}+1$ dimensional space. And the inter-user interference can be eliminated as
\begin{equation}
\begin{aligned}
\mathbb{E}\left[ \phi_{\mathrm{UL},k}^\mathrm{ZF'}\right] 
&=\sum_{i\neq k}p_{\mathrm{UL}}\mathbb{E}\left| \frac{\hat{\mathbf{a}}_{\mathrm{UL},k}^\mathrm{H}}{\left\| \hat{\mathbf{a}}_{\mathrm{UL},k}\right\|}\mathbf{A}\hat{\mathbf{g}}_{\mathrm{UL},i}\right|^2=0.
\end{aligned}
\end{equation}
Others are similar to MRC and we omit the proof of ZFR due to the limitation of space.
\end{proof}
\end{appendices}
\vspace{-0.25cm}
\bibliographystyle{IEEEtran}
\bibliography{IEEEabrv,mybib}
\vspace{-2cm}
\begin{IEEEbiography}[{\includegraphics[width=1in,height=1.25in,clip,keepaspectratio]{ 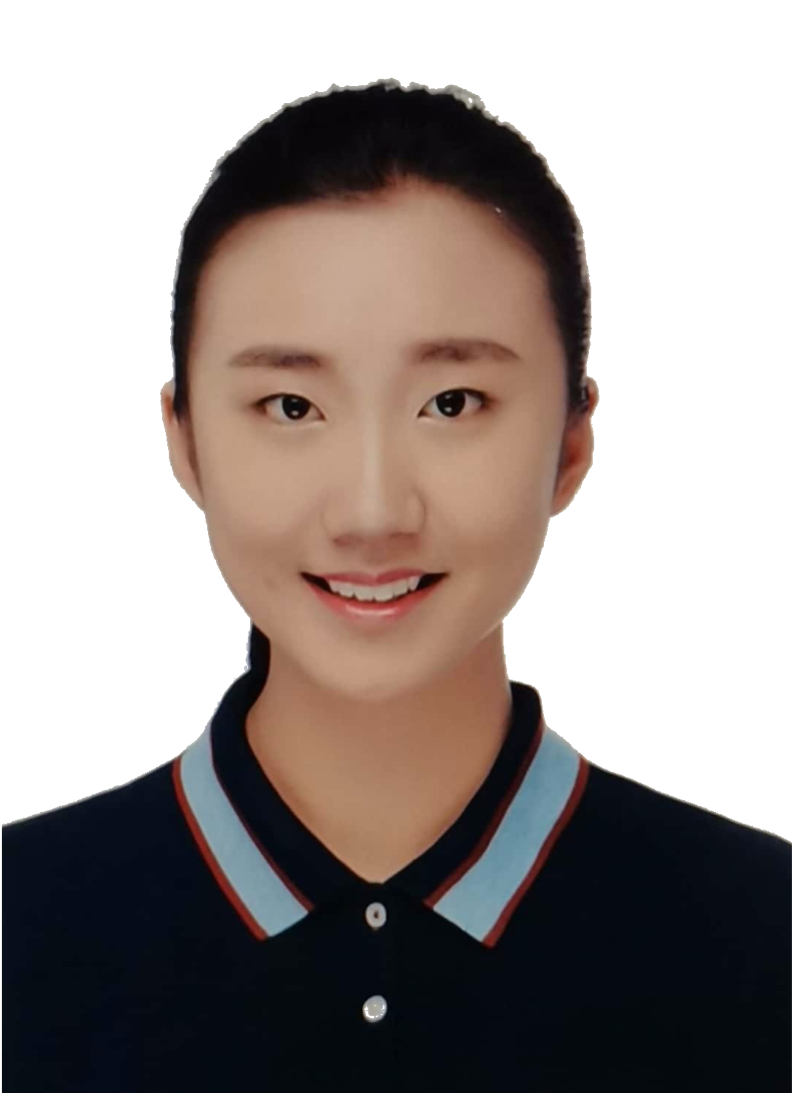}}]{Xiangning Song}was born in Shandong province, China, in 1999. She received the B.S. degree (\emph{summa cum laude}) in information engineering from Southeast University, Nanjing, China, in 2021. She is currently pursuing the master degree in communication and information system at the National Mobile Communications Research Laboratory, Southeast University. Her research interests include distributed massive MIMO, low-resolution ADCs, and cooperative communications.\\
\end{IEEEbiography}
\vspace{-1.75cm}
\begin{IEEEbiography}[{\includegraphics[width=1in,height=1.25in,clip,keepaspectratio]{ 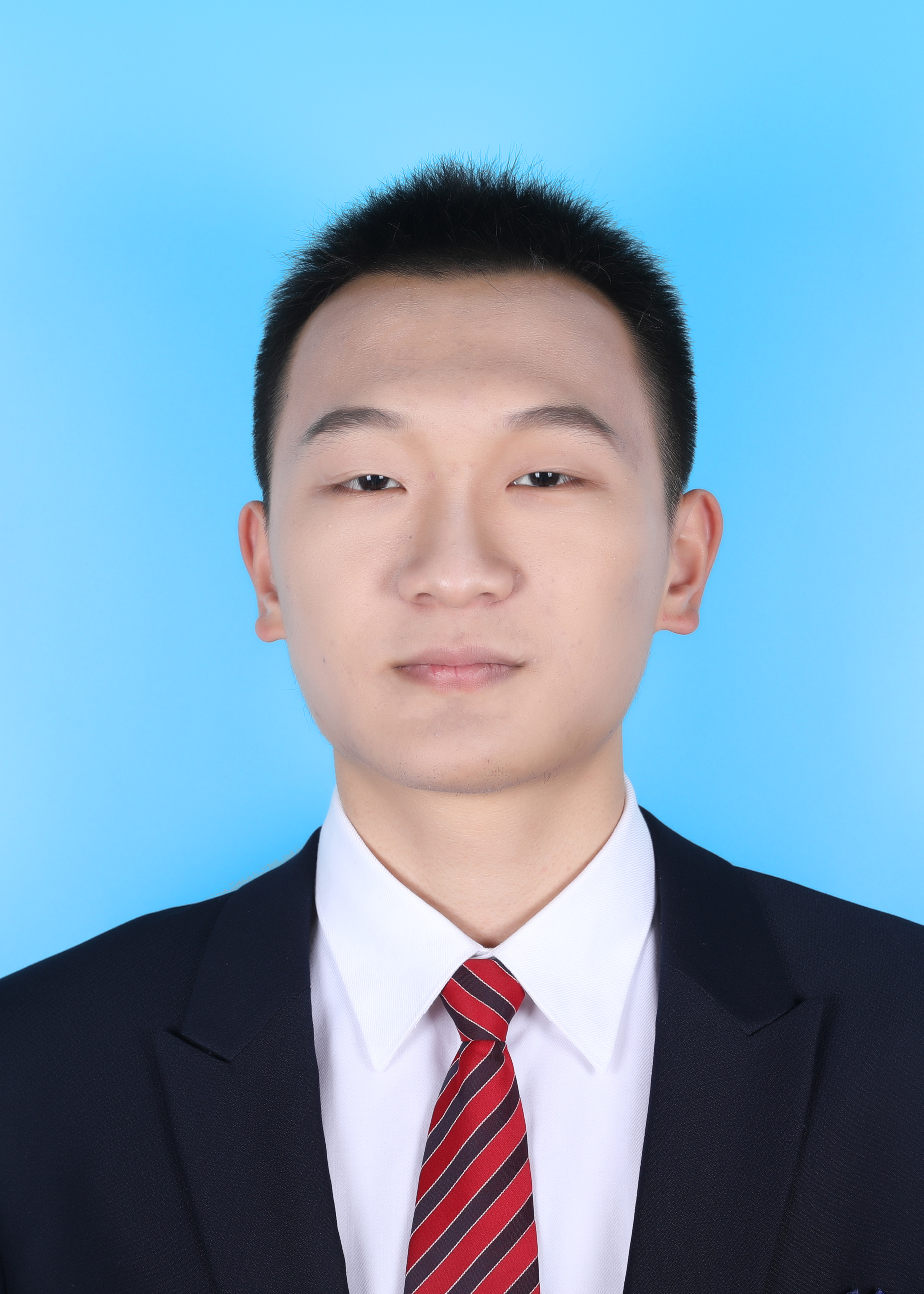}}]{Zhenhao Ji}
(S'19) received the B.S. degree in information engineering from Southeast University, Nanjing, China, in 2021. He is currently pursuing the M.S. degree in communication system at the National Mobile Communications Research Laboratory, Southeast University. He was the World Winner of 2019-2020 IEEE CASS Student Design Competition. His research interests includes AI in communication and VLSI design for digital signal processing.
\end{IEEEbiography}
\vspace{-1.75cm}
\begin{IEEEbiography}[{\includegraphics[width=1in,height=1.25in,clip,keepaspectratio]{ 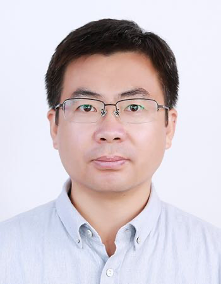}}]{Jiamin Li} received the B.S. and M.S. degrees in communication and information systems from Hohai University, Nanjing, China, in 2006 and 2009, respectively, and the Ph.D. degree in information and communication engineering from Southeast University, Nanjing, China, in 2014. He joined the National Mobile Communications Research Laboratory, Southeast University, in 2014, where he has been an Associate Professor since 2019. His research interests include cell-free distributed massive MIMO, massive ultra-reliable low-latency communications (mURLLC), artificial intelligence and its applications in future mobile communications.\\
\end{IEEEbiography}
\vspace{-1.75cm}
\begin{IEEEbiography}[{\includegraphics[width=1in,height=1.25in,clip,keepaspectratio]{ 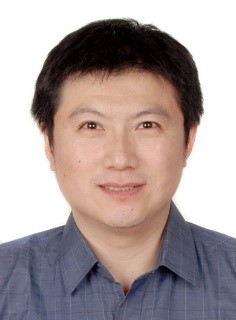}}]{Pengcheng Zhu} received the B.S and M.S. degrees in electrical engineering from Shandong University, Jinan, China, in 2001 and 2004, respectively, and the Ph.D. degree in communication and information science from the Southeast University, Nanjing, China, in 2009. He has been a lecturer with the national mobile communications research laboratory, Southeast University, China. His research interests lie in the areas of communication and signal processing, including limited feedback techniques, and distributed antenna systems.\\
\end{IEEEbiography}
\vspace{-1.75cm}
\begin{IEEEbiography}[{\includegraphics[width=1in,height=1.25in,clip,keepaspectratio]{ 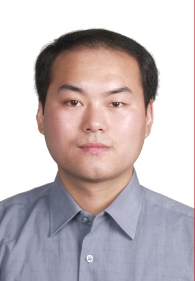}}]{Dongming Wang} (M'06) received the B.S. degree from Chongqing University of Posts and Telecommunications, Chongqing, China, the M.S. degree from Nanjing University of Posts and Telecommunications, Nanjing, China, and the Ph.D. degree from the Southeast University, Nanjing, China, in 1999, 2002, and 2006, respectively. He joined the National Mobile Communications Research Laboratory, Southeast University, in 2006, where he has been an Associate Professor since 2010. His research interests include turbo detection, channel estimation, distributed antenna systems, and large-scale MIMO systems.
\end{IEEEbiography}
\vspace{-1.75cm}
\begin{IEEEbiography}[{\includegraphics[width=1in,height=1.25in,clip,keepaspectratio]{ 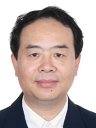}}]{Xiaohu You} (F'11) received the B.S., M.S. and Ph.D. degrees in electrical engineering from Nanjing Institute of Technology, Nanjing, China, in 1982, 1985, and 1989, respectively. From 1987 to 1989, he was with Nanjing Institute of Technology as a Lecturer. From 1990 to the present time, he has been with Southeast University, first as an Associate Professor and later as a Professor. His research interests include mobile communications, adaptive signal processing, and artificial neural networks with applications to communications and biomedical engineering. He is the Chief of the Technical Group of China 3G/B3G Mobile Communication R \& D Project. He received the excellent paper prize from the China Institute of Communications in 1987 and the Elite Outstanding Young Teacher Awards from Southeast University in 1990, 1991, and 1993. He was also a recipient of the 1989 Young Teacher Award of Fok Ying Tung Education Foundation, State Education Commission of China.
\end{IEEEbiography}

\end{document}